\documentclass[journal]{IEEEtran}

\usepackage{graphicx}
\usepackage{amsmath}
\usepackage[noend]{algpseudocode}
\usepackage{algorithmicx,algorithm}
\usepackage{amsthm}
\usepackage{amsfonts}
\usepackage{subfigure} 
\usepackage{color}
\usepackage{cite}
\usepackage{setspace}
\newtheorem{Lemma}{Lemma}
\newtheorem{Theorem}{Theorem}
\newtheorem{Definition}{Definition}
\newtheorem{Corollary}{Corollary}
\newtheorem{Proposition}{Proposition}
\newtheorem{Assumption}{Assumption}
\usepackage{amssymb}
\usepackage{color}
\graphicspath{{images/}}
\usepackage{booktabs}
\usepackage{multirow} 
\usepackage{makecell}
 \usepackage[numbers,sort&compress]{natbib}
 \newtheorem{remark}{Remark}
\usepackage{stfloats}
\usepackage{amsmath}
\usepackage{amssymb}
\usepackage{enumitem}
\usepackage{threeparttable}
\usepackage{comment}
\usepackage{bm}
\usepackage{amssymb}        % 提供 \checkmark (打钩) 符号
\usepackage{booktabs}       % 提供 \toprule, \midrule, \bottomrule (三线表)
\usepackage{makecell}       % 提供 \makecell 命令 (单元格内换行)
\usepackage[table]{xcolor}  % 提供表格颜色支持 (gray!10 背景色)
\usepackage{tabularx}
\usepackage{newtxtext}
\usepackage{newtxmath}
\begin{document}	
\title{Fluid-Spatiotemporal Stochastic Geometry: Information Flow in Non-Stationary Fields
}
	
	\author{Wen-Yu Dong, Weiwei Jiang,~\IEEEmembership{Senior Member,~IEEE}, Song Zhao,\\ Qi Bi,~\IEEEmembership{Fellow,~IEEE}, Sheng Chen,~\IEEEmembership{Life Fellow,~IEEE}	%
	
	\thanks{W.-Y. Dong, S. Zhao and Q. Bi are with Future Technology Research Center, China Telecom Research Institute, Beijing 102209, China (E-mails: dongwy@chinatelecom.cn; zhaosong1@chinatelecom.cn; qibi@chinatelecom.cn)} %
	\thanks{W. Jiang is with the School of Information and Communication Engineering, Beijing University of Posts and Telecommunications, Beijing, 100876, China (Email: jww@bupt.edu.cn)}
	\thanks{S. Chen is with the School of Electronics and Computer Science, University of Southampton, Southampton SO17 1BJ, U.K., and also with Faculty of Information Science and Technology, Ocean University of China, Qingdao 266100, China (E-mail: sqc@ecs.soton.ac.uk).} %
	\vspace*{-5mm}
}

\maketitle 

\begin{abstract}
	The fundamental limits of information flow in spatial networks have been extensively characterized under the assumption of stationary spatial point processes.
	However, this stationarity hypothesis fails to capture the macroscopic transport of information demand in regimes where the node intensity field exhibits continuous, non-separable spatiotemporal evolution.
	This paper establishes the theoretical foundations of Fluid-Spatiotemporal Stochastic Geometry (F-STSG), treating the dynamic topology as a hydrodynamic limit of the discrete node constellation.
	We formulate the identification of latent network dynamics as an inverse boundary value problem.
	By invoking the principle of minimum kinetic energy consistent with Optimal Transport theory, we prove the existence and uniqueness of a scalar potential field that strictly governs the compressive evolution of the network load.
	This field-theoretic formulation establishes a rigorous field-measure coupling between the continuous Lagrangian transport and the discrete Eulerian interference geometry.
	Based on this, we derive the Information Flux vector, a sufficient statistic for the macroscopic advection of the capacity region, and establish the Material Derivative as the kinematic predictor of topological divergences.
	Finally, we characterize the fundamental limits of such non-stationary systems through two key theoretical contributions.
	{\color{black}First, by analyzing the trade-off between spectral efficiency and topological coordination overhead, we derive an asymptotic scaling law for the optimal node density under a quadratic-overhead approximation.}
	We prove that in the interference-limited regime, the energy-optimal topology scales as the square root of the structural cost ratio, defining a thermodynamic inverse-square information barrier independent of link-level spectral efficiency.
	{\color{black}Second, we reveal a fundamental source-channel duality perspective for mobile networks. By proving that the macroscopic flow divergence mathematically determines the topological entropy production rate, we establish the physical foundation for the information-theoretic cost of mobility. This demonstrates that tracking the dynamic network state requires control signaling capacities that scale fundamentally with the kinematic entropy of the topology.}
\end{abstract}


\section{Introduction}\label{sec:intro} % S1

\IEEEPARstart{C}{haracterizing} the fundamental limits of information flow in stochastic fields is a central problem in network information theory.
Classical analyses, from Shannon's capacity to modern stochastic geometry (SG) \cite{1995Stochastic, Baccelli2009}, are predominantly predicated on the axiom of stationarity.
By modeling the spatial configuration of nodes as a realization of a stationary point process, typically a Poisson point process (PPP), these frameworks leverage the ergodic hypothesis to equate time averages with spatial ensemble averages.
While this approach has successfully derived asymptotic scaling laws for static networks \cite{Gupta2000, Haenggi2012}, it encounters a fundamental theoretical barrier when applied to non-stationary regimes.

In operational scenarios, the spatial configuration of users is rarely static or purely random. Instead, it exhibits continuous, collective, and non-separable spatiotemporal evolution.
We identify this macroscopic regime of collective mobility as the ``Digital Tide''.
Unlike microscopic Brownian motion where node displacements are independent, the Digital Tide represents a structured, compressible flow of information demand, whose underlying stochastic process is inherently non-ergodic.
Consequently, the classical notion of ergodic capacity becomes ill-defined, and the system's performance limits are no longer governed by static parameters, but by the hydrodynamic transport properties of the information field itself.

Within this context, existing SG frameworks reveal a critical limitation: they typically treat network evolution merely as a sequence of independent realizations. This quasi-static approximation implicitly discards the temporal causality inherent in the transport dynamics. In non-stationary regimes driven by exogenous rhythms, the Shannon capacity limit transforms from a static parameter into a dynamic manifold governed by the trajectory of the distribution itself. Standard metrics, therefore, fail to predict topological divergences, such as the formation of transient hotspots, before they manifest as outages.

To rigorously capture these dynamics, we analyze the network in the thermodynamic limit where the number of nodes $N \to \infty$ \cite{Lasry2007}.
In this asymptotic regime, the discrete node constellation converges to a deterministic continuum field $\lambda(\bm{x},t)$ defined on the spatiotemporal domain $\mathbb{R}^d \times \mathbb{R}^+$---the mathematical manifestation of the Digital Tide.
This approach allows us to apply field-theoretic tools to quantify the collective transport of information.
However, a fundamental theoretical gap exists in this hydrodynamic limit: there is currently no analytical framework capable of coupling the continuous Lagrangian dynamics of the node intensity field with the discrete Eulerian statistics of network interference.
Existing methodologies track mass transport but fail to map it rigorously to the information-theoretic performance of the infrastructure.

This paper bridges this gap by establishing the theoretical foundations of Fluid-Spatiotemporal Stochastic Geometry (F-STSG).
Departing from heuristic mobility models, we formulate the identification of latent network dynamics as an inverse boundary value problem~\cite{Isakov2006}.
Our primary objective is to recover the unique, minimum-energy macroscopic velocity field $\bm{v}(\bm{x},t)$ solely from the observable density evolution.
This formulation unifies continuum mechanics with point process theory, enabling the characterization of the fundamental kinematic, thermodynamic, and information-theoretic limits of non-stationary networks.

\subsection{Related Works}\label{sec:related_works} % S1.1

Modeling the kinetics of large-scale networks presents a fundamental trade-off between physical fidelity and analytical tractability.
To contextualize the theoretical landscape, this section categorizes existing literature into three distinct epistemological approaches: the microscopic snapshots of SG, the macroscopic flows of continuum models, and the abstract representations of data-driven and graph models.
We analyze the structural limitations of each paradigm regarding spatiotemporal consistency and conclude by identifying the fundamental information-theoretic constraints, particularly the thermodynamic cost of mobility, that remain unaddressed in current network control theory.

\subsubsection{Microscopic Stochastic Geometry: From Static Fields to Reactive Deployment}
Following the foundational works in \cite{Baccelli2009} and \cite{Haenggi2012}, SG has become a standard analytical tool for spatial averaging \cite{DongTCOM, DongJSAC, DongIoTJ, DongGC, DongGC2}. However, this framework relies primarily on the separation of space and time domains. Specifically, the spatial distribution is typically modeled as a stationary process, serving as a fixed background upon which temporal variations are analyzed.

This structural separation creates a fundamental dichotomy in reliability analysis, particularly concerning temporal interference correlation. 
Because the spatial topology is treated as a fixed realization, the interference experienced by a user becomes highly correlated across consecutive time slots. 
For instance, the study in \cite{Haenggi2013} identifies a phase transition in reliability, revealing that such static interference fields can lead to infinite local delay, as a user suffering from strong interference is likely to face the same blockage in subsequent retransmissions. 
Consequently, existing investigations are typically confined to limiting regimes: either quasi-static interference where the topology remains frozen, or fast-varying interference \cite{Lu2021, Win2009} where the topology is independently redrawn in every slot. In these scenarios, as well as in retransmission analysis \cite{Nigam2015, Krishnan2017}, the continuous spatiotemporal evolution of the network is absent.

To address these static or memoryless extremes, subsequent research attempted to explicitly incorporate node mobility, yet these extensions often fail to capture macroscopic structural evolution.
Existing mobility models primarily examine how individual node displacement modifies the local interference field rather than modeling the collective evolution of the user density function \cite{Krishnan2017, Tabassum2019}. 
This limitation parallels studies on transmitter activity via queueing theory \cite{Yang2018, Zhong2018}: both approaches are predicated on a time-invariant background intensity, implying that the partial derivative of density with respect to time is zero. 
Essentially, these frameworks track the microscopic motion of particles without capturing the macroscopic flow of the medium.

The limitations of such microscopic analytical lens constrain network synthesis and control, particularly in dynamic orchestration strategies such as unmanned aerial vehicle deployments.
Lacking a predictive kinematic law for density evolution, current literature typically addresses demand through a localized approach where deployments serve as reactive measures to maintain local stability. 
For instance, resources are dispatched only when specific clusters experience performance degradation \cite{Zhang2022, Al-Hourani2016}. 
These methods rely on sequential static snapshots, effectively precluding the analysis of macroscopic advection phenomena in which the topology exhibits non-separable and continuous spatiotemporal evolution.

\subsubsection{Macroscopic Flow Models: The Inverse Problem and Rationality}
To address the unrealistic assumption of static constraints in SG, approximating dense networks as continuous media provides a dynamic alternative rooted in vehicular traffic flow theory \cite{Lighthill1955, Greenshields1935} and pedestrian dynamics \cite{Helbing1995, Hughes2002}. Within the wireless domain, mean field games (MFG) \cite{Lasry2007} adopt similar continuum principles to model edge caching \cite{Kim2020}, computation offloading \cite{Zheng2021}, and interference management \cite{Zhang2019, Wang2014}.

A fundamental limitation of these frameworks lies in their formulation as a forward problem, which postulates future density evolution based on prescribed microscopic drivers. 
Classical paradigms, most notably the Lighthill-Whitham-Richards (LWR) models \cite{Lighthill1955}, depend on empirical fundamental diagrams such as velocity-density constitutive laws \cite{Greenshields1935, Hughes2002}. 
Similarly, MFG relies on the assumption of rational utility maximization, including the minimization of energy cost \cite{Zheng2021} or security risks \cite{Wang2014}. 
In realistic environments, however, these drivers are typically latent or inaccessible, as the subjective utility functions governing human mobility are fundamentally opaque to external observers. 
A theoretical gap exists for an inverse problem formulation capable of reconstructing macroscopic dynamics solely from observable aggregate data, thereby circumventing reliance on unverifiable microscopic behavioral assumptions.

This reliance on complex behavioral coupling results in significant mathematical intractability. The structure of MFG typically involves a system of coupled forward-backward non-linear partial differential equations (PDEs), comprising the Hamilton-Jacobi-Bellman and Fokker-Planck-Kolmogorov equations \cite{Zhang2019, Kim2020}. Such systems rarely admit analytical solutions and necessitate computationally expensive iterative numerical solvers. 
Unlike linear systems, these coupled non-linear structures preclude the derivation of explicit closed-form scaling laws, limiting theoretical insight into how network performance scales with fundamental mobility parameters.

Beyond these modeling and computational constraints, classical fluid models exhibit a geometric discrepancy regarding field-measure decoupling. These approaches track mass transport as a continuous field without explicitly mapping the density to the discrete interference geometry governing the signal-to-interference-plus-noise ratio (SINR). 
To date, a unified framework that rigorously couples continuous mass transport with the discrete random measures governing interference statistics remains absent.

\subsubsection{Abstract and Data-Driven Frameworks: The Loss of Physical Causality}
Alternative frameworks typically prioritize computational or analytical tractability at the expense of spatial resolution or physical interpretability. 
Temporal point processes \cite{Gonzalez2016, Farrahi2014} efficiently capture event arrival rates along the temporal dimension but impose a one-dimensional abstraction. 
This formulation inherently relies on a distance-agnostic interference assumption, effectively reducing network interactions to mean-field intensities. 
Consequently, the distinct spatial clustering required for rigorous SINR analysis is obscured, rendering these models insufficient for geometry-dependent interference management.

In the domain of complex network dynamics \cite{Boccaletti2006, Barabasi2013}, the analysis shifts focus to topological graph properties such as degree distribution. 
By formulating the network as a topological graph $G(V,E)$ defined on the sets of nodes $V$ and edges $E$, rather than a geometric constellation in $\mathbb{R}^d$, these models abstract away the Euclidean metric space. 
This abstraction precludes the derivation of path-loss dependent metrics including area spectral efficiency (ASE), as the critical relationship between physical distance and signal attenuation is discarded.

From a computational perspective, data-driven paradigms leveraging machine learning \cite{Wang2023, Ding2025, Li2024} approach traffic prediction through high-dimensional optimization. 
Although these methods effectively identify statistical correlations, they generally function as phenomenological models that lack explicit physical derivation. 
Currently, a theoretical framework that derives macroscopic network evolution directly from first principles, thereby reconciling data-driven prediction with deterministic physical interpretability, remains to be established.

\subsubsection{Information Theoretic Limits: Mobility as Cost vs. Gain}
Beyond specific modeling techniques, the analysis must address the information-theoretic constraints imposed by network dynamics.
The foundational work by Grossglauser and Tse \cite{Grossglauser2022} demonstrates that node mobility can theoretically increase the capacity scaling of ad hoc networks to scale linearly with $N$, denoted as $\Theta(N)$, by utilizing nodes as physical relays.
However, this capacity gain is strictly predicated on the assumption of ergodic mixing over extended time scales.
As quantified by the work of \cite{Gamal2006}, achieving such throughput benefits necessitates a large delay tolerance, which strictly bounds the delay-throughput tradeoff.
Consequently, for latency-sensitive applications, the network topology acts as a real-time constraint rather than a delay-tolerant resource, rendering the relaying gain inaccessible.

In such real-time regimes, topological dynamics function not as a relaying mechanism, but as a primary source of estimation and control costs.
At the physical link level, this cost manifests as channel aging, where Doppler shifts reduce the channel coherence time and render feedback information rapidly obsolete.
While the scaling laws of channel state information (CSI) feedback are established for static links \cite{Jindal2006}, the penalty induced by this continuous state degradation \cite{Truong2013} remains under-characterized in network-level analysis.
Specifically, current frameworks lack a unified metric to quantify how the aggregate information loss from mobility offsets the spatial multiplexing gains from densification.

At the macroscopic system level, tracking a continuously evolving topology is information-theoretically analogous to stabilizing an unstable plant over a noisy channel.
According to the data rate theorem \cite{Tatikonda2004}, the minimum information rate required to maintain bounded estimation error is lower-bounded by the topological entropy of the system.
Yet, existing literature has not formulated a theoretical boundary that explicitly connects this entropic cost of mobility to the potential gains of network capacity.
Crucially, a rigorous mapping is missing that translates the physical kinetics of node movement into the information-theoretic cost required to track it.
Establishing such a thermodynamic cost function is essential to formulate a source-channel duality for mobile network control, yet this remains an open problem in the literature.

\subsection{Key Theoretical Questions: Searching for Fundamental Limits}\label{S1.2}

The transition from stationary snapshots to the hydrodynamic limit of SG fundamentally alters the problem space.
In this non-stationary regime, the classical capacity region is no longer a static polygon but a dynamic manifold driven by the information demand field.
To explore the fundamental limits of such systems, three critical questions arise regarding the kinematics, thermodynamics, and informatics of the network flow:

\begin{enumerate}
	\item \textbf{The Kinematic Limit (Minimum Energy Transport Bound):} The scalar continuity equation represents a fundamental conservation law but is mathematically underdetermined for the vector velocity field.
	\textit{Does there exist a unique, minimum-energy macroscopic velocity field that explains the observed density evolution, thereby establishing a kinematic lower bound on the transport cost of the network load?}
	
	{\color{black}\item \textbf{The Thermodynamic Limit (Energy-Capacity Scaling):}
	In the interference-limited regime, densification increases ASE but incurs super-linear coordination costs.  \textit{Does a structural invariant exist that defines the thermodynamic limit of network densification—specifically, a universal scaling law for the optimal node density that is asymptotically independent of physical layer parameters such as spectral efficiency?} }

	\item \textbf{The Information-Theoretic Limit (Source-Channel Duality):}
	As the network topology evolves, it acts as a non-stationary information source.
	\textit{Is there a fundamental information-theoretic lower bound on the control signaling rate required to track this dynamic state? Can we prove that the macroscopic flow divergence is mathematically equivalent to the entropy production rate of the topology?}
\end{enumerate}

\begin{table*}[t]
	\small
	\centering
	\caption{Methodological Comparison and Uniqueness of F-STSG}
	\label{tab:gap_analysis} % Tab.I
	\vspace*{-2mm}
	\renewcommand{\arraystretch}{1.5} 
	\setlength{\tabcolsep}{6pt} 
	\begin{threeparttable}
		\begin{tabular}{l c c c >{\columncolor{gray!10}}c} 
			\toprule
			\multicolumn{1}{c}{\textbf{Theoretical Dimension}} & 
			\textbf{\makecell{Stochastic\\Geometry (SG)}} & 
			\textbf{\makecell{Continuum Flow\\(MFG/LWR)}} & 
			\textbf{\makecell{Graph \& Data\\Driven Models}} & 
			\textbf{\makecell{F-STSG\\(Proposed)}} \\
			\midrule
			
			% Row 1: 几何精度
			\textbf{1. Euclidean Spatial Resolution} & \checkmark & $\times$ & $\times$ & \textbf{\checkmark} \\
			\textit{ (Explicitly models Distance, Path Loss, \& SINR)} & & & & \\
			
			% Row 2: 动态特性
			\textbf{2. Non-Stationary Evolution} & $\times$ & \checkmark & \checkmark & \textbf{\checkmark} \\
			\textit{ (Captures time-varying densities $\partial_t \lambda \neq 0$)} & & & & \\
			
			% Row 3: 逆向问题 (关键修改点)
			\textbf{3. Inverse Problem Tractability} & $\times$ & $\times$ & $\times$ & \textbf{\checkmark} \\
			\textit{ (Recovers vector field from aggregate counts)} & & & & \\
			
			% Row 4: 行为假设依赖 (攻击 MFG 的痛点)
			\textbf{4. Independence from Micro-Drivers} & $\times$ & $\times$ & \checkmark & \textbf{\checkmark} \\
			\textit{ (Avoids latent utility functions or rational agents)} & & & & \\
			
			% Row 5: 解析解
			\textbf{5. Closed-Form Scaling Laws} & \checkmark & $\times$ & $\times$ & \textbf{\checkmark} \\
			\textit{ (Analytically derivable vs. Numerical/Iterative)} & & & & \\
			
			% Row 6: 场-测度耦合
			\textbf{6. Discrete-Continuum Coupling} & $\times$ & $\times$ & $\times$ & \textbf{\checkmark} \\
			\textit{ (Bridges fluid transport with discrete interference)} & & & & \\
			
			% Row 7: 热力学/信息论
			\textbf{7. Thermodynamic Consistency} & $\times$ & $\times$ & $\times$ & \textbf{\checkmark} \\
			\textit{ (Quantifies Entropic Cost of Mobility)} & & & & \\
			
			\midrule
			\multicolumn{1}{c}{\textbf{Key Limitations Identified}} & 
			\makecell[t]{\textit{Static /}\\ \textit{Time-Slotted}} & 
			\makecell[t]{\textit{Math Intractability /}\\ \textit{Subjective Utility}} &  % 加上 /
			\makecell[t]{\textit{Geometry-Blind /}\\ \textit{Lack Causality}} &  % 这里也建议加上 /
			-- \\
			\midrule
			\multicolumn{1}{c}{\textbf{Relevant References}} & 
			\makecell[t]{\cite{Baccelli2009, Haenggi2013}} & 
			\makecell[t]{\cite{Lasry2007, Kim2020,Lighthill1955}} &  % 把 LWR 加在这里
			\makecell[t]{\cite{Boccaletti2006, Wang2023}} & 
			\textbf{This Work} \\
			\bottomrule
		\end{tabular}
		\begin{tablenotes}
			\footnotesize
			\item MFG: Mean Field Games, LWR: Lighthill-Whitham-Richards, F-STSG: Fluid-Spatiotemporal Stochastic Geometry. \\ \checkmark: considered, $\times$ : not considered, -- : indicates the aforementioned limitations are resolved.
		\end{tablenotes}
	\end{threeparttable}
	\vspace*{-3mm}
\end{table*}

\subsection{Main Contributions}\label{S1.3}

This paper addresses these questions by establishing F-STSG, a field-theoretic framework capable of rigorously deriving macroscopic flow dynamics directly from the temporal evolution of node intensity.
The primary contributions are summarized as follows:

\begin{itemize}
\item \textbf{Resolution of the Inverse Problem via Optimal Transport:}
We formulate the identification of latent network dynamics as an inverse boundary value problem. By invoking the principle of minimum kinetic energy, consistent with the theory of Optimal Transport \cite{Benamou2000}, we justify the imposition of an irrotational constraint. We prove the existence and uniqueness of a macroscopic velocity field that strictly governs the compressive evolution of the network load, resolving the ill-posedness of the scalar continuity equation.

	\item \textbf{Lagrangian Information Kinetics:}
	We generalize microscopic node mobility to the Information Flux, a vector field quantifying the spatial advection of capacity demand.
	Bridging continuum mechanics and Shannon theory, we establish that the flux vector constitutes a sufficient statistic for the dynamic capacity region.
	Furthermore, we derive the Material Derivative as the kinematic predictor of topological divergences, generalizing the concept of outage probability to non-stationary hydrodynamic regimes.
	
	{\color{black}\item \textbf{Asymptotic Energy-Capacity Scaling Laws under Quadratic Overhead:}
	We characterize the asymptotic limits of fluid networks. Grounded in the degrees-of-freedom analysis of interference channels where feedback overhead scales quadratically with density \cite{Jindal2006}, we derive an asymptotic closed-form scaling law for the optimal node density.}
	We prove that in the interference-limited regime, the energy-optimal topology scales as the square root of the structural cost ratio ($\lambda^* \propto \sqrt{\mathcal{P}_{\mathrm{static}}/\kappa}$, where $\lambda^*$ is the energy-optimal node density, $\mathcal{P}_{\mathrm{static}}$ is the baseline static power consumption, and $\kappa$ is the effective coordination coefficient.).
	This result constitutes a structural invariant, defining a thermodynamic barrier independent of link-level spectral efficiency (SE).
	
	\item \textbf{The Source-Channel Duality of Mobility:}
{\color{black}	We elevate the analysis from physical transport to a continuous information-theoretic perspective. By modeling the dynamic topology as an evolving information source, we prove that the macroscopic velocity divergence mathematically determines the topological entropy production rate. This establishes the physical foundation for the source-channel duality of network mobility, where the divergence of the Information Flux serves as an intrinsic kinematic footprint of the signaling complexity required to track the non-stationary network state.}
\end{itemize}

Table~\ref{tab:gap_analysis} provides a systematic comparison between the proposed F-STSG framework and existing modeling paradigms, highlighting the key methodological gaps addressed by our approach. The remainder of this paper is organized as follows.

Section~\ref{sec:system_model} establishes the F-STSG system model.
It rigorously defines the ``Digital Tide'' as the hydrodynamic limit of the mobile node constellation and introduces the dynamic network state process, which mathematically couples the continuous intensity field with the discrete, stochastic infrastructure measure.

Section~\ref{sec:core_theory} presents the framework's core theoretical engine: the reconstruction of the macroscopic velocity field $\bm{v}(\bm{x},t)$.
We address the fundamental under-determination of the continuity equation by invoking the theory of Optimal Transport.
Specifically, we identify the physical velocity field as the unique minimizer of the kinetic energy functional (Proposition~\ref{Prop_Optimal_Transport}), which justifies the irrotational constraint.
This transforms the kinematic reconstruction into a well-posed inverse boundary value problem for the scalar flow potential (Theorem~\ref{Theorem1}).

Section~\ref{sec:linear_response} generalizes this solution to arbitrary, non-symmetric topologies.
We develop a Green's function formalism using perturbation theory, deriving the linearized flow potential equation (Theorem~\ref{Theorem2}) and discussing the spectral properties of the network transport operator.

Section~\ref{sec:metrics_toolbox} constructs the novel field-theoretic analytical toolbox.
We derive two complementary sets of indicators:
(i) Lagrangian Kinematics, including the Information Flux vector (Proposition~\ref{Prop_Flux_Integral}), the Material Derivative as a stability statistic (Proposition~\ref{Prop_Material_Derivative}), and the Centroid Drift Velocity (Proposition~\ref{Prop_Centroid_Drift}); and
(ii) Dynamic Performance Metrics, specifically mapping the field evolution to the instantaneous Association and Coverage Probabilities (Propositions~\ref{Prop_Assoc_Prob} and \ref{Prop_Cov_Prob}) via spatial expectations.

Section~\ref{sec:simulation} provides rigorous numerical corroboration.
By comparing the macroscopic analytical predictions against microscopic event-driven Monte Carlo simulations, we verify the asymptotic exactness of the hydrodynamic limit assumption and the validity of the field-measure coupling.

Section~\ref{sec:applications} explores the fundamental theoretical insights and asymptotic limits of the framework.
The analysis culminates in two key contributions:
{\color{black}(i) the derivation of an asymptotic energy-capacity scaling law (Proposition~\ref{Prop_Scaling}), revealing that the optimal node density scales as the square root of the structural cost ratio ($\lambda^* \propto \sqrt{\mathcal{P}_{\mathrm{static}}/\kappa}$);} and
(ii) the establishment of a fundamental source-channel duality, proving that the Information Flux divergence constitutes the exact topological entropy production rate (Proposition~\ref{Prop_Entropy}).

Section~\ref{sec:variants} discusses the theoretical generalizations of the framework.
We formulate rigorous mathematical extensions for handling non-conservative dynamics (via non-homogeneous Poisson equations), rotational vorticity (via Helmholtz decomposition), and structural correlations in the infrastructure (via Pair Correlation Functions (PCFs)).

Section~\ref{sec:future_work} outlines promising avenues for future research, including the rate-distortion theory of topology and physics-informed deep learning.
Finally, Section~\ref{sec:conclusion} concludes the paper.

\subsection{Notations}\label{S1.4}

Throughout this paper, $\mathbb{R}^d$ denotes the $d$-dimensional Euclidean space, and $\mathcal{B}(\mathbb{R}^d)$ denotes the Borel $\sigma$-algebra on $\mathbb{R}^d$. 
Boldface lower-case letters denote vectors, e.g., $\bm{x}$. $\|\bm{x}\|$ denotes the Euclidean norm of vector $\bm{x}$, while $\|f\|_{L^2}$ denotes the $L^2$-norm of function $f$.
For a bounded domain $\Omega \subset \mathbb{R}^d$, $\partial \Omega$ denotes its boundary.
For spatiotemporal field $f(\bm{x},t)$ and vector field $\bm{F}(\bm{x},t)$, the operators $\nabla f$, $\nabla \cdot \bm{F}$, $\nabla \times \bm{F}$, and $\Delta f$ (or $\nabla^2 f$) denote the gradient, divergence, curl, and Laplacian, respectively.
The material derivative is denoted by $\frac{D}{Dt}\! \triangleq\! \frac{\partial}{\partial t} + \bm{v} \cdot \nabla$. The notation $\mu_N\! \xrightarrow{w}\! \lambda$ indicates the weak convergence of measure $\mu_N$ to density $\lambda$.
$\mathbb{E}[\cdot]$ and $\mathbb{P}(\cdot)$ denote statistical expectation and probability, respectively. 
The Dirac measure concentrated at $\bm{x}$ is denoted by $\delta_{\bm{x}}$, and $\mathbb{I}(\cdot)$ represents the indicator function. 
The reduced Palm distribution of a point process $\Phi$ is denoted by $\mathbb{P}^!_{\Phi}$. 
Standard asymptotic notations are used, e.g., $f(x) \sim g(x)$ implies $\lim_{x \to \infty} f(x)/g(x) = 1$, $f(x) = \mathcal{O}(g(x))$ implies $\limsup_{x \to \infty} |f(x)/g(x)| < \infty$, and $f(x) \in \Theta(g(x))$ indicates that $g(x)$ is an asymptotically tight bound for $f(x)$.
Key system-specific symbols and variables used throughout the paper are summarized in Table~\ref{tab:key_notations}.

\begin{table}[tbp]
	\small
	\caption{Summary of Key Notations}
	\label{tab:key_notations} % Tab.II
	\vspace*{-2mm}
	\centering
	\renewcommand{\arraystretch}{1.25}
	\begin{tabularx}{\columnwidth}{ll}
		\toprule
		\textbf{Symbol} & \textbf{Definition} \\
		\midrule
		\multicolumn{2}{l}{\textit{Network Model \& Physical Layer}} \\
		$\Phi_{\mathrm{B}}$ & Base station (BS) point process \\
		$\Psi(t)$ & Set of active BSs at time $t$ \\
		$\bm{X}_i, \mathcal{K}(i)$ & Location and tier index of the $i$-th BS \\
		$P_k, \mathcal{W}_k$ & Transmit power and bias weight of tier $k$ \\
		$\ell_k(r), \alpha_k$ & Path loss function and exponent of tier $k$ \\
		$\gamma_k$ & SINR threshold for tier $k$ \\
		$I(\bm{x},t)$ & Aggregate interference field \\
		$\sigma^2, h$ & Noise power and small-scale fading gain \\
		
		\multicolumn{2}{l}{\textit{Fluid-Spatiotemporal Dynamics}} \\
		$\lambda(\bm{x},t)$ & Macroscopic user intensity (``Digital Tide'') \\
		$N(t)$ & Total mobile user population \\
		$\bm{v}(\bm{x},t)$ & Macroscopic velocity field \\
		$\bm{J}(\bm{x},t)$ & Information flux vector ($\bm{J} \triangleq \lambda \bm{v}$) \\
		$\mathcal{D}(\bm{x},t)$ & Congestion divergence ($\mathcal{D} \triangleq \nabla \cdot \bm{v}$) \\
		$\phi(\bm{x},t)$ & Scalar potential (Irrotational component $\bm{v}_{\mathrm{irr}}$) \\
		$\bm{A}(\bm{x},t)$ & Vector potential (Rotational component $\bm{v}_{\mathrm{rot}}$) \\
		$S(\bm{x},t)$ & Source/generation rate in continuity equation \\
		$\bm{C}(t)$ & Network load centroid \\
		
		\multicolumn{2}{l}{\textit{Perturbation \& Analysis}} \\
		$\lambda_0, \phi_0$ & Zeroth-order background intensity and potential \\
		$\delta \lambda, \delta \phi$ & First-order intensity fluctuation, potential perturbation \\
		$G(\bm{x},\bm{y})$ & Green's function for the transport operator \\
		$\psi_n, \nu_n$ & Eigenfunctions and eigenvalues of transport operator \\
		$\mathcal{L}_0$ & Linearized transport operator $\nabla \cdot (\lambda_0 \nabla)$ \\
		
		\multicolumn{2}{l}{\textit{Performance Metrics \& Extensions}} \\
		$P_{\mathrm{cov}}(t)$ & Instantaneous coverage probability \\
		$A_k(t)$ & Tier-$k$ association probability \\
		$H(t)$ & Topological differential entropy \\
		$\Omega(\lambda)$ & Network energy cost functional \\
		$\kappa$ & Coordination cost coefficient \\
		$\lambda^*$ & Energy-optimal node density \\
		$\bar{R}_{\infty}$ & Asymptotic ergodic spectral efficiency \\
		$g(r)$ & Pair correlation function (PCF) \\
		\bottomrule
	\end{tabularx}
	\vspace*{-3mm}
\end{table}

%======================================================================
%=====================Section II=======================================
%======================================================================
\section{System Model}\label{sec:system_model} % S2

We introduce a general analytical framework, termed F-STSG. By adopting a field-theoretic perspective, this framework treats the discrete node constellation as a continuous compressible fluid,  constructed to characterize the fundamental limits of non-stationary networks where the mobile node topology exhibits macroscopic, collective evolution. As illustrated in Fig.~\ref{fig:system_model}, the framework is formulated as the coupling of two distinct structures: a static discrete random measure representing the network infrastructure, and a continuous intensity field representing the hydrodynamic limit of the information demand.

\subsection{Network Skeleton: Discrete Random Measures}\label{sec:model_skeleton} % S2.1

We model the network infrastructure as a discrete random measure $\Phi_{\mathrm{B}}$ on the Euclidean domain $\mathbb{R}^d$~\cite{Daley2003}:
\begin{equation} % eq.1
	\Phi_{\mathrm{B}} = \sum_{i} \delta_{\bm{X}_i},
\end{equation}
where $\{\bm{X}_i\}$ denote the random locations of the infrastructure nodes,e.g., base stations (BSs), and $\delta_{\bm{x}}$ represents the Dirac measure concentrated at $\bm{x}$. While our subsequent numerical validation assumes a PPP for analytical tractability in the interference derivation, the core kinematic theory developed in Section~\ref{sec:core_theory} holds for any simple, stationary ergodic point process.
In the specific case of a multi-tier heterogeneous network, $\Phi_{\mathrm{B}}$ is the superposition of $K$ independent point processes, $\Phi_{\mathrm{B}} = \bigcup_{k=1}^K \Phi_k$. A mobile node at location $\bm{x} \in \mathbb{R}^d$ associates with the infrastructure according to a stationary policy based on the maximum biased received power. 

\begin{figure}[!t]
	\centering
	\includegraphics[width=0.93\columnwidth]{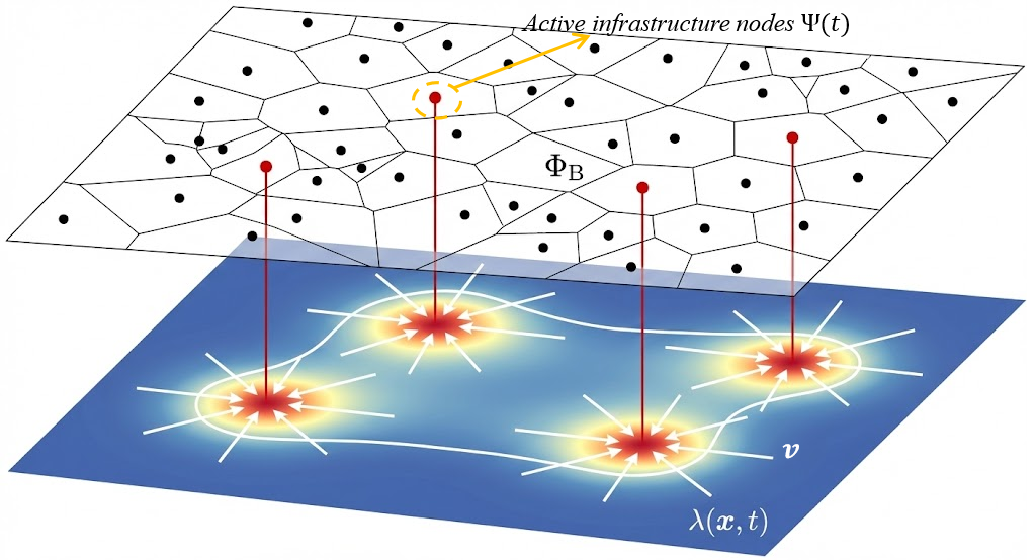}
	\vspace*{-2mm}
	\caption{Illustration of the Field-Measure Coupling in F-STSG.}
	\label{fig:system_model} % Fig.1
	\vspace*{-4mm}
\end{figure}

To formalize the association regions, let $\mathcal{K}(n)$ denote the index of the tier to which node $\bm{X}_n$ belongs. We denote the path loss function for tier $k$ as $\ell_k(\cdot): \mathbb{R}^+ \to \mathbb{R}^+$. While the F-STSG framework supports general monotonic path loss models, we adopt the standard power-law function $\ell_k(r) = r^{-\alpha_k}$ in the subsequent analysis for tractability, where $\alpha_k > 2$ is the path loss exponent.
We consider the weighted Voronoi tessellation where the association cell $V_i$ of a generic node $\bm{X}_i$ is defined as:
\begin{align} % eq.2
	V_i = &\left\{\bm{x} \in \mathbb{R}^d : P_{\mathcal{K}(i)} \mathcal{W}_{\mathcal{K}(i)} \ell_{\mathcal{K}(i)}(\|\bm{x} - \bm{X}_i\|) \ge \right.\nonumber\\
	&\quad\,\, \left. P_{\mathcal{K}(j)} \mathcal{W}_{\mathcal{K}(j)} \ell_{\mathcal{K}(j)}(\|\bm{x} - \bm{X}_j\|), \forall j \neq i \right\},
	\label{eq:assoc_rule_generic}
\end{align}
where $P_k$ and $\mathcal{W}_k$ denote the transmit power and association bias weight for tier $k$, respectively.

It is necessary to strictly distinguish the modeling of the infrastructure from that of the mobile nodes. While we apply a hydrodynamic limit to the mobile node population, as detailed in Section~\ref{sec:model_user_population}, we explicitly retain the discrete nature of the infrastructure $\Phi_{\mathrm{B}}$. This distinction prevents the loss of higher-order statistical moments—such as the spatial clustering of interference—which are typically smoothed out in pure mean-field models where BSs are approximated as a continuous density. The F-STSG framework thus mathematically couples a continuous intensity field representing demand with a discrete random measure representing geometry.

\subsection{Continuum Limit: Thermodynamic Limit and Digital Tide}\label{sec:model_user_population} % S2.2

A central feature of this work is the rigorous transition from microscopic mobility to macroscopic flow dynamics. We consider the distribution of mobile nodes—acting as information sources or sinks—as a dynamic stochastic system. To bridge the gap between discrete entities and continuous dynamics, we model the system in the thermodynamic limit.
Let $\{\Phi_{\mathrm{M}}^{(N)}(t)\}_{N \ge 1}$ be a sequence of point processes representing the locations of mobile nodes, indexed by the expected population size $N$. Mathematically, this is modeled as a sequence of Cox processes driven by a deterministic intensity for any Borel set $A \in \mathcal{B}(\mathbb{R}^d)$. Let $\mu_N(\cdot, t)$ be the empirical measure associated with the $N$-th process:
\begin{equation} % eq.3
	\mu_N(A, t) = \frac{1}{N} \sum_{\bm{Y}_j \in \Phi_{\mathrm{M}}^{(N)}(t)} \mathbb{I}(\bm{Y}_j \in A), \quad \forall A \in \mathcal{B}(\mathbb{R}^d).
\end{equation}
We assume that as $N \to \infty$, the normalized empirical measure converges weakly to a deterministic probability density function (PDF) $p(\bm{x},t)$:
\begin{equation} % eq.4
	\mu_N(\cdot, t) \xrightarrow{w} p(\bm{x},t) \mathrm{d}\bm{x} \triangleq \frac{\lambda(\bm{x},t)}{N(t)} \mathrm{d}\bm{x} ,
	\label{eq:weak_convergence}
\end{equation}
where $N(t)\! =\! \int_{\mathbb{R}^d} \lambda(\bm{x},t) \,\mathrm{d}\bm{x}$ is the total expected population, and $\lambda(\bm{x},t)$ is the macroscopic intensity field.
This asymptotic regime justifies the use of smooth PDEs to describe the evolution of discrete node constellations.

\begin{Definition}\label{def:digital_tide} % Def.1
	The Digital Tide is formally defined as the time-varying intensity field $\lambda(\bm{x},t) : \mathbb{R}^d \times \mathbb{R}^+ \to \mathbb{R}^+$, representing the hydrodynamic limit of the mobile node process. The discrete realization at any instant $t$, $\Phi_{\mathrm{M}}(t)$, is modeled as an inhomogeneous PPP (IPPP) governed by this intensity, such that for any Borel set $A \subseteq \mathbb{R}^d$, the expected number of nodes is:
	\begin{equation} % eq.5
		\mathbb{E}[N_A(t)] = \int_A \lambda(\bm{x},t) \,\mathrm{d}\bm{x}.
	\end{equation}
	Crucially, the evolution of the Digital Tide is strictly governed by the continuity equation of mass conservation, driven by a macroscopic velocity field $\bm{v}(\bm{x},t)$. This formulation assumes that microscopic node trajectories are continuous in time, ensuring that local density changes are exclusively attributable to flux divergence.
\end{Definition}

\begin{remark} % Rem.1
	We adopt the term ``Digital Tide'' to emphasize the non-separable spatiotemporal nature of the demand field, distinct from stationary breathing models. The continuum approximation in \eqref{eq:weak_convergence} is asymptotically exact in the regime of ultra-dense networks (UDN). In this regime, the relative stochastic fluctuation (shot noise) scales as $\mathcal{O}(1/\sqrt{\lambda})$, rendering the hydrodynamic description a statistical necessity for characterizing system stability.
\end{remark}

%----------new lemma-  Displacement theory-------------------
To rigorously justify the IPPP assumption in Definition \ref{def:digital_tide} under continuous spatiotemporal evolution, we must bridge the deterministic fluid kinematics with the stochastic properties of the discrete measure. This field-measure isomorphism is established via the fundamental displacement theorem of point processes.

\begin{Lemma}\label{Lemma_Displacement}[Topological Invariance under Kinematic Advection]
	Let the initial discrete node constellation $\Phi_{\mathrm{M}}(0)$ be a Poisson point process with intensity measure $\lambda_0(\bm{x})$. If the constituent nodes undergo independent deterministic displacements over time $t$, governed by the macroscopic velocity field $\bm{v}(\bm{x},t)$, the advected point process $\Phi_{\mathrm{M}}(t)$ remains strictly a Poisson point process for all $t > 0$. Furthermore, its intensity measure $\lambda(\bm{x},t)$ is exactly the unique solution to the conservative continuity equation $\partial_t \lambda + \nabla \cdot (\lambda \bm{v}) = 0$.
\end{Lemma}

\begin{proof}
	This follows directly from the Mapping Theorem (or Displacement Theorem) of Poisson point processes \cite{Baccelli2009}. Let $\mathcal{T}_t: \mathbb{R}^d \to \mathbb{R}^d$ denote the deterministic flow map generated by the velocity field $\bm{v}(\bm{x},t)$, such that the trajectory of a node $\bm{x}_i(t)$ satisfies $\frac{\mathrm{d}\bm{x}_i}{\mathrm{d}t} = \bm{v}(\bm{x}_i, t)$. Since the mapping $\mathcal{T}_t$ applies independently to each point in $\Phi_{\mathrm{M}}(0)$ and ensures no atom accumulation (under a smooth velocity field), the mapped process $\Phi_{\mathrm{M}}(t) = \{\mathcal{T}_t(\bm{x}_i) : \bm{x}_i \in \Phi_{\mathrm{M}}(0)\}$ preserves complete spatial randomness. The conservation of probability measure along the Lagrangian trajectories guarantees that the transformed intensity $\lambda(\bm{x},t)$ rigorously satisfies the continuity equation.
\end{proof}

\begin{remark}%[Field-Measure Isomorphism and PGFL Restart]
	Lemma \ref{Lemma_Displacement} provides the indispensable mathematical foundation for the F-STSG framework. It proves that solving the macroscopic partial differential equation (PDE) for fluid density is structurally equivalent to deriving the intensity measure of the evolved stochastic geometry. Consequently, the physical continuous fluid density is legally isomorphic to the probabilistic spatial Poisson intensity. This isomorphism allows us to legitimately restart the probability generating functional (PGFL) and Laplace functional analysis at any non-stationary snapshot $t$, utilizing the PDE solution $\lambda(\bm{x},t)$ as the exact integral kernel for microscopic interference computation.
\end{remark}

\begin{remark}%Robustness to Microscopic Turbulence and Stochastic Kinematics
	The preservation of the PPP property established in Lemma~\ref{Lemma_Displacement} is not limited to strictly deterministic flow fields. According to the Independent Random Translation theorem for Poisson processes, if individual nodes exhibit microscopic stochastic deviations (e.g., Brownian turbulence $\bm{v}'$ introduced in Section~\ref{sec:turbulence}) superimposed on the macroscopic drift $\bar{\bm{v}}$, the evolved network remains strictly a PPP, provided these deviations are mutually independent. In such turbulent regimes, the continuity equation governing the intensity measure $\lambda(\bm{x},t)$ naturally generalizes to a Fokker-Planck (advection-diffusion) equation, flawlessly maintaining the field-measure isomorphism without breaking the analytical tractability of the Laplace functional.
\end{remark}
\subsection{Topological Dynamics: Non-Separable Intensity Fields}\label{sec:model_engine} % S2.3

The structural properties of the intensity field $\lambda(\bm{x},t)$ dictate the complexity of the network dynamics. We distinguish between two fundamental classes based on the coupling between spatial geometry and temporal evolution.

\begin{Definition} % Def.2
	A spatiotemporal field is separable if it admits the multiplicative factorization $\lambda(\bm{x},t) = f(\bm{x})g(t)$. Mathematically, this implies that the spatial gradient of the temporal logarithmic derivative vanishes identically over the entire domain:
	\begin{equation} % eq.6
		\nabla_{\bm{x}} \left( \frac{\partial}{\partial t} \ln \lambda(\bm{x},t) \right) \equiv \bm{0}, \quad \forall (\bm{x},t) \in \mathbb{R}^d \times \mathbb{R}^+.
		\label{eq:sep_condition}
	\end{equation}
	Separable fields model networks with static hotspots that merely undergo uniform amplitude scaling (global pulsation) without spatial displacement or deformation.
\end{Definition}

\begin{Definition} % Def.3
	A field is non-separable if the spatial structure and temporal evolution are inextricably coupled. Formally, we require that the mixed derivative is not identically zero:
	\begin{equation} % eq.7
		\nabla_{\bm{x}} \left( \frac{\partial}{\partial t} \ln \lambda(\bm{x},t) \right) \not\equiv \bm{0}.
		\label{eq:non_sep_condition}
	\end{equation}
	{\color{black}This condition mathematically guarantees that the intensity profile undergoes continuous spatial deformation or translation, serving as a rigorous indicator of macroscopic topological advection rather than mere local amplitude scaling.}
\end{Definition}

To illustrate this classification, consider the canonical time-varying power-law topology, often referred to as the ``Breathing City'' model. The intensity is modeled as a radially symmetric field governed by a mass-conserving power-law decay profile with a time-varying exponent $\beta(t)$:
\begin{equation} % eq.8
	\lambda(r,t) = \Lambda_0(t) (d_0 + r)^{-\beta(t)}, \quad r = \|\bm{x}\|,
\end{equation}
where  $d_0 > 0$ is a regularization parameter to ensure finiteness at the origin, and $\Lambda_0(t)$ is a time-dependent normalization factor uniquely determined by the global mass conservation constraint $\int \lambda(\bm{x},t) \,\mathrm{d}\bm{x} = N(t)$. Applying the non-separability criterion derived above, the mixed derivative yields:
\begin{equation} % eq.9
	\nabla_{\bm{x}} \frac{\partial}{\partial t} \ln \lambda \!=\!  -\dot{\beta}(t) \nabla_{\bm{x}} \ln(d_0 + \|\bm{x}\|)\! =\! - \frac{\dot{\beta}(t)}{d_0 + \|\bm{x}\|} \frac{\bm{x}}{\|\bm{x}\|}.
\end{equation}
This vector field is non-vanishing for any dynamic shape evolution $\dot{\beta}(t) \neq 0$. This confirms that a time-varying spatial decay generates a rigorous macroscopic velocity field, necessitating the field-theoretic approach developed in Section~\ref{sec:core_theory}.

\subsection{Dynamic State Process (Fluid-Infrastructure Coupling)}\label{sec:model_coupling} % S2.4

A distinguishing feature of the F-STSG framework is its ability to model the interaction between the continuous intensity field $\lambda(\bm{x},t)$ and the discrete network infrastructure $\Phi_{\mathrm{B}}$. This interaction is modeled via the dynamic network state process.

Let $\Psi(t) \subseteq \Phi_{\mathrm{B}}$ denote the subset of active infrastructure nodes at time $t$. To capture heterogeneous dynamics, we decompose this set into tier-specific active subsets $\Psi(t) = \bigcup_{k=1}^K \Psi_k(t)$, where $\Psi_k(t) \subseteq \Phi_k$. The activation of on-demand nodes is modeled as a functional of the underlying fluid field. Specifically, for a node in tier $k$, the local activation probability $p_{a,k}(\bm{x},t)$ is driven by the local fluid intensity $\lambda(\bm{x},t)$. This establishes a field-measure coupling, where the continuous demand field governs the discrete infrastructure realization.
Given a realization of the active set $\Psi(t)$, a mobile node at $\bm{x}$ associates with the network based on the maximum biased received power. Consistent with the notation in Section~\ref{sec:model_skeleton}, let $\mathcal{K}(i)$ denote the tier index of node $\bm{X}_i$. The serving node $\bm{X}^*$ is given by:
\begin{equation} % eq.10
	\bm{X}^* = \arg \max_{\bm{X}_i \in \Psi(t)} \left( P_{\mathcal{K}(i)} \mathcal{W}_{\mathcal{K}(i)} \ell_{\mathcal{K}(i)}(\|\bm{x} - \bm{X}_i\|) \right),
	\label{eq:assoc_rule}
\end{equation}
where $P_k$, $\mathcal{W}_k$, and $\ell_k(r)$ denote the transmit power, association bias weight, and path loss for tier $k$, respectively. This general rule captures load balancing and tier prioritization within the dynamic topology.

\subsection{Instantaneous Capacity Field and Information Outage}\label{sec:model_phy} % S2.5

The coupling between the continuous mobile node distribution and the discrete infrastructure induces a dynamic Shannon capacity field, denoted as $C(\bm{x},t)$. Consider a typical mobile node located at $\bm{x}_0 \in \mathbb{R}^d$ at time $t$. Conditioned on the active infrastructure realization $\Psi(t)$ and the association result $\bm{X}^*$, the instantaneous channel capacity (in bits/s/Hz) is given by:
\begin{equation} % eq.11
	C(\bm{x}_0, t) = \log_2(1 + \text{SINR}(\bm{x}_0, t)).
	\label{eq:capacity_def}
\end{equation}
Here, the SINR acts as the interface between the network geometry and information theory. Let $k^* = \mathcal{K}(\bm{X}^*)$ denote the index of the serving tier. The SINR is defined as:
\begin{equation} % eq.12
	\text{SINR}(\bm{x}_0, t) = \frac{P_{k^*} h_{\bm{X}^*} \ell_{k^*}(\|\bm{x}_0 - \bm{X}^*\|)}{I(\bm{x}_0, t) + \sigma^2},
	\label{eq:sinr_def}
\end{equation}
where $h_{\bm{X}^*}$ represents the random channel power gain (e.g., Rayleigh fading), $\sigma^2$ is the additive noise power, and $I(\bm{x}_0, t)$ is the aggregate interference field generated by all other active nodes:
\begin{equation} % eq.13
	I(\bm{x}_0, t) = \sum_{\bm{X}_j \in \Psi(t) \setminus \{\bm{X}^*\}} P_{\mathcal{K}(j)} h_j \ell_{\mathcal{K}(j)}(\|\bm{x}_0 - \bm{X}_j\|).
\end{equation}
Note that the summation is performed over $\Psi(t) = \bigcup_{k=1}^K \Psi_k(t)$, rigorously capturing inter-tier interference dynamics.

The primary performance metric is the instantaneous network coverage probability, $P_{\mathrm{cov}}(t)$. Unlike static analysis where user locations are typically assumed to be uniform, here the receiver location $\bm{x}$ is drawn from the time-varying intensity field $\lambda(\bm{x},t)$.
Mathematically, under the adiabatic assumption (see Section~\ref{sec:network_metrics}), $P_{\mathrm{cov}}(t)$ is defined as the spatial average of the conditional probability that the SINR exceeds the tier-specific threshold. Let $\gamma_k$ denote the SINR threshold for tier $k$, corresponding to a rate target $R_k = \log_2(1+\gamma_k)$. The coverage probability is:
\begin{equation} % eq.14
	P_{\mathrm{cov}}(t) \!\triangleq \!\mathbb{E}_{\bm{x} \sim p(\cdot,t)} \!\left[ \mathbb{P}_{h}\left( \text{SINR}(\bm{x}, t) > \gamma_{\mathcal{K}(\bm{X}^*)} \mid \Psi(t) \! \right)\! \right],
	\label{eq:pcov_def_rigorous}
\end{equation}
where the inner probability $\mathbb{P}_{h}(\cdot)$ is taken over the realizations of the small-scale fading gains $\{h_i\}$, conditioned on the active node locations $\Psi(t)$ with the normalized spatial PDF of the mobile nodes given by $p(\bm{x},t) = \lambda(\bm{x},t) / \int \lambda \,\mathrm{d}\bm{x}$. This definition explicitly captures the field-measure coupling: the network performance is weighted by the local demand density rather than the geometric area.

%======================================================================
%=====================Section III=======================================
%======================================================================
\section{Field Reconstruction: The Inverse Boundary Value Problem}\label{sec:core_theory} % S3

In the hydrodynamic limit established in Section~\ref{sec:system_model}, the network topology is fully characterized by the continuous intensity field $\lambda(\bm{x},t)$. While $\lambda(\bm{x},t)$ is the observable state, the underlying driver of network evolution is the macroscopic transport of demand. The theoretical challenge addressed in this section is the identification of these latent transport dynamics. We formulate this as a mathematical inverse problem: reconstructing the vector velocity field $\bm{v}(\bm{x},t)$ solely from the scalar density evolution $\lambda(\bm{x},t)$.

\vspace*{-2mm}
\subsection{Conservation Laws and Underdetermined System}\label{sec:continuity_equation} % S3.1
\vspace*{-1mm}

The fundamental link between topological evolution and transport dynamics is governed by the conservation of mass. For a conservative node constellation characterized by a vanishing source term $S(\bm{x},t)$, i.e., $S(\bm{x},t)=0$, where $S(\bm{x},t)$ represents the local rate of demand generation or depletion, the local rate of change in intensity is balanced strictly by the divergence of the flux vector:
\begin{equation} % eq.15
	\frac{\partial \lambda(\bm{x},t)}{\partial t} + \nabla \cdot (\lambda(\bm{x},t) \bm{v}(\bm{x},t)) = 0.
	\label{eq:continuity_conservative}
\end{equation}

Assuming that the per-node information demand is statistically homogenous, the conservation of node mass implies the conservation of aggregate information demand. Thus, $\lambda(\bm{x},t)$ serves as the macroscopic descriptor for the information density state.

The application of this continuum law to node kinematics is not a mere heuristic analogy but a rigorous mathematical consequence of the thermodynamic limit established in Section~\ref{sec:model_user_population}. Crucially, since individual mobile nodes follow continuous trajectories without discontinuous jumps, i.e., no teleportation, the macroscopic evolution of the intensity field must satisfy mass conservation. Eq.~\eqref{eq:continuity_conservative} does not imply that nodes interact like physical fluid particles via mechanisms such as viscosity; rather, it asserts that the temporal variation of the local node intensity is mathematically equivalent to the divergence of a transport flux. Thus, the fluid description is a statistical necessity of large-scale, continuous mobility.

It is recognized that total load dynamics are driven by two distinct mechanisms: local generation via birth-death processes ($S \neq 0$) and spatial advection via mobility-induced transport ($-\nabla \cdot (\lambda \bm{v})$). The F-STSG framework deliberately focuses on the conservative regime where $S=0$, based on the principle of mechanistic decoupling:
\begin{itemize}
	\item Generation dynamics constitute a scalar, local modulation of intensity, primarily impacting capacity dimensioning.
	\item Advection dynamics constitute a vector, non-local reconfiguration of the topology.
\end{itemize}
By isolating the conservative component, we strictly characterize the transport-induced limits of the network—specifically the Information Flux—distinct from static queuing stability.

Mathematically, \eqref{eq:continuity_conservative} represents a single scalar constraint on the $d$-dimensional vector field $\bm{v}(\bm{x},t)$. Consequently, the problem of determining $\bm{v}$ from $\lambda$ is ill-posed, as the solution is non-unique up to an arbitrary solenoidal field where a vector field $\bm{u}$ satisfying $\nabla \cdot (\lambda \bm{u}) = 0$.

\begin{remark} % Rem.2
	It is important to differentiate the F-STSG framework from classical macroscopic mobility models found in vehicular traffic flow theory, exemplified by the LWR model.
	\begin{enumerate}
		\item \textbf{Inverse vs. Forward Formulation}: Traffic theory typically solves a forward problem by prescribing a constitutive relation between velocity and density (e.g., a fundamental diagram where velocity decreases with congestion) to predict density evolution. In contrast, F-STSG solves the inverse problem by treating the non-separable intensity evolution $\lambda(\bm{x},t)$ as the observable input to reconstruct the latent minimum-energy velocity field necessary to satisfy conservation laws.
		\item \textbf{Field-Measure Coupling}: Traffic models typically focus solely on the hydrodynamics of the fluid. F-STSG, in contrast, is fundamentally concerned with the coupling between this continuous fluid representing information demand and the discrete static random measure modeling the network infrastructure. This coupling allows for the derivation of information-theoretic metrics, such as SINR coverage and topological entropy, which are inherently absent in pure flow mechanics.
	\end{enumerate}
\end{remark}

To render this inverse problem well-posed, we require a rigorous regularization condition based on the kinematics of the Digital Tide.

\subsection{Regularization via Optimal Transport Theory}\label{sec:problem_formulation} % S3.2

Solving the conservative continuity equation \eqref{eq:continuity_conservative} for the vector field $\bm{v}(\bm{x},t)$ poses a fundamental mathematical challenge: the equation is underdetermined. The single scalar constraint imposed by mass conservation is insufficient to uniquely define a $d$-dimensional vector field.
To resolve this ill-posedness rigorously, we seek the specific solution that is physically most meaningful.

\subsubsection{Helmholtz Decomposition and Inverse Problem}
According to the Helmholtz decomposition theorem  \cite{Batchelor1967}, any sufficiently smooth vector field $\bm{v}(\bm{x},t)$ can be decomposed into a curl-free component and a divergence-free component:
\begin{equation} % eq.16
	\bm{v} = \bm{v}_{\mathrm{irr}} + \bm{v}_{\mathrm{rot}} = -\nabla \phi + \nabla \times \bm{A} .
	\label{eq:helmholtz}
\end{equation}
where $\phi$ denotes the scalar potential (associated with the irrotational component $\bm{v}_{\mathrm{irr}}$), $\bm{A}$ denotes the vector potential (associated with the solenoidal component $\bm{v}_{\mathrm{rot}}$), and $\nabla \times$ represents the curl operator.

Substituting this into the continuity equation $\partial_t \lambda + \nabla \cdot (\lambda \bm{v}) = 0$ yields:
\begin{equation} % eq.17
	\frac{\partial \lambda}{\partial t} = \nabla \cdot (\lambda \nabla \phi) - \nabla \cdot (\lambda \bm{v}_{\mathrm{rot}}).
	\label{eq:expanded_continuity}
\end{equation}
The term $\nabla \cdot (\lambda \bm{v}_{\mathrm{rot}}) = \bm{v}_{\mathrm{rot}} \cdot \nabla \lambda$ represents advective mixing. While this term can contribute to local density variations, it does not represent the net transport of mass required to resolve the source-sink ambiguity. {\color{black}The continuity equation admits infinite solutions. To select the unique physical solution and resolve this ill-posedness, we formulate a constrained variational problem. Specifically, we isolate the component responsible for the efficient compressive transport of the ``Digital Tide'' by invoking the principle of minimum instantaneous kinetic energy.}

\subsubsection{The Minimum-Energy Solution}
We isolate the specific component responsible for the efficient compressive transport of the ``Digital Tide''.

\begin{Proposition}\label{Prop_Optimal_Transport} %[Minimum Energy Transport] % Pro.1
	Among all possible velocity fields $\bm{v}$ that satisfy the continuity equation for a given density evolution $\partial_t \lambda$, the irrotational field $\bm{v}^* = -\nabla \phi$ is the unique solution that minimizes the instantaneous kinetic energy of the network flow, defined as:
	\begin{equation} % eq.18
		E(t) = \int_{\mathbb{R}^d} \lambda(\bm{x},t) \|\bm{v}(\bm{x},t)\|^2 \,\mathrm{d}\bm{x}.
	\end{equation}
\end{Proposition}

\begin{proof}
See Appendix~\ref{app:proof_prop1}.
\end{proof}

	\begin{remark}%[Physical Interpretation and Scope of Applicability] % Rem.3
		It is crucial to distinguish the reconstructed velocity field $\hat{\bm{v}}$ from the microscopic mobility traces of individual users. 
		The actual trajectory of a user may contain significant rotational components (vorticity) or random meanderings. 
		The F-STSG framework projects this complex reality onto a canonical minimum-energy manifold. 
		This projection implies a macroscopic filtering mechanism validated by the Helmholtz decomposition:
		\begin{itemize}
			\item \textit{Filtration of Stationarity:} If the real rotation is topologically stationary, i.e., movement along density isolines ($\bm{v}_{\mathrm{rot}} \perp \nabla \lambda$), it consumes kinetic energy but does not alter the capacity demand distribution. F-STSG filters this out as kinematic noise.
			\item \textit{Equivalent Transport:} If the real rotation induces advection (transporting mass across gradients), F-STSG captures the net mass transport and represents it as an equivalent irrotational flux $\hat{\bm{v}}$ that generates the identical density evolution $\partial_t \lambda$.
		\end{itemize}
		Therefore, $\hat{\bm{v}}$ represents the thermodynamic lower bound of the transport cost. The framework is strictly applicable for macroscopic resource orchestration---where the objective is to quantify the net transfer of capacity demand---rather than for microscopic user tracking.
	\end{remark}

Substituting the optimal solution $\bm{v} = -\nabla \phi$ into \eqref{eq:continuity_conservative}, we derive the governing Poisson-like PDE for the flow potential:
\begin{equation} % eq.19
	\frac{\partial \lambda(\bm{x},t)}{\partial t} = \nabla \cdot \left( \lambda(\bm{x},t) \nabla \phi(\bm{x},t) \right).
	\label{eq:conservative_pde}
\end{equation}
This equation forms the theoretical core of our framework, establishing a deterministic link between the observable density evolution and the latent scalar potential.

%========================================================
\subsection{Existence and Uniqueness of Scalar Potential}\label{sec:analytical_solution} % S3.3

The flow potential equation \eqref{eq:conservative_pde} derived above constitutes the governing law for the ``Digital Tide'' in the conservative regime characterized by a vanishing source term $S(\bm{x},t)=0$. This equation represents a linear elliptic PDE for the potential $\phi(\bm{x},t)$ at each time instant $t$. We now establish the mathematical well-posedness of this field reconstruction problem.

{\color{black}While the fundamental infrastructure point processes are formulated on the whole space $\mathbb{R}^d$ to avoid boundary edge effects in interference analysis, practical network orchestration and fluid tracking operate within finite geographical boundaries (e.g., a metropolitan area or a specific control zone). To rigorously establish the well-posedness of the flow reconstruction, we restrict our kinematic analysis to a bounded operational domain where the user demand does not completely vanish.
}

\begin{Theorem}\label{Theorem1} %[Existence and Uniqueness]
	Let $\Omega \subset \mathbb{R}^d$ be a bounded domain with smooth boundary $\partial \Omega$, and $\lambda(\bm{x},t)$ be a strictly positive, smooth spatiotemporal intensity field satisfying the global mass conservation condition $\frac{\mathrm{d}}{\mathrm{d}t}\int_{\Omega}\lambda \,\mathrm{d}\bm{x} = 0$.
	Consider the Neumann boundary value problem for the scalar potential $\phi(\bm{x},t)$:
	\begin{align} % eqs.20,21
		\nabla \cdot (\lambda(\bm{x},t) \nabla \phi(\bm{x},t)) &= \frac{\partial \lambda(\bm{x},t)}{\partial t} \quad \text{in } \Omega, \label{eq:poisson_general} \\
		\frac{\partial \phi}{\partial n} &= 0 \quad \text{on } \partial \Omega. \label{eq:neumann_bc}
	\end{align}
	The boundary condition \eqref{eq:neumann_bc} imposes a no-flux constraint ($\bm{v} \cdot \mathbf{n} = 0$), ensuring that the system is closed. Under these conditions, there exists a solution $\phi(\bm{x},t)$ that is unique up to an additive constant. Consequently, the macroscopic velocity field $\bm{v} = -\nabla \phi$ is uniquely determined.
\end{Theorem}

\begin{proof}
	Equation \eqref{eq:poisson_general} is of the form $\mathcal{L}\phi = f$, where $\mathcal{L} = \nabla \cdot (\lambda \nabla)$ is a self-adjoint, uniformly elliptic operator given the condition $\lambda(\bm{x},t) > 0$.
	According to the Fredholm alternative~\cite{Evans2010}, a solution exists if and only if the source term $f = \partial_t \lambda$ is orthogonal to the kernel of the adjoint operator $\mathcal{L}^*$ (which consists of constant functions for Neumann boundaries).
	The solvability condition is thus:
	\begin{equation} % eq.22
		\int_{\Omega} \frac{\partial \lambda(\bm{x},t)}{\partial t} \,\mathrm{d}\bm{x} = \frac{\mathrm{d}}{\mathrm{d}t} \int_{\Omega} \lambda(\bm{x},t) \,\mathrm{d}\bm{x} = 0.
	\end{equation}
	This condition is strictly satisfied by the global conservation of mass assumption. Uniqueness up to an additive constant follows from the maximum principle.
\end{proof}

{\color{black}Although Theorem~\ref{Theorem1} is established for strictly positive fields on bounded domains, the framework accommodates whole-space models on $\mathbb{R}^d$. For spatial profiles where $\lambda \to 0$ as $\|\bm{x}\| \to \infty$, the governing operator $\nabla \cdot (\lambda \nabla)$ exhibits degenerate ellipticity. This singularity is analytically resolved by formulating the problem in the weighted Sobolev space $W^{1,2}_\lambda(\mathbb{R}^d)$. Under the constraint of finite macroscopic kinetic energy, this space guarantees the existence and uniqueness (up to a constant) of the potential $\phi$, consistently bridging the localized fluid kinematics with global point process formulations.}

%While Theorem \ref{Theorem1} guarantees existence for arbitrary topologies, analytical tractability often requires symmetry.

\begin{Corollary}\label{Cor_Symmetric}  % [Canonical Solution for Radially Symmetric Systems] % Cor.1
	For the canonical case of a radially symmetric intensity field $\lambda(r,t)$ on a disk of radius $R_{\max}$, the unique macroscopic radial velocity field $v_r(r,t)$ at radius $r$ admits the following closed-form integral representation:
	\begin{equation} % eq.23
		v_r(r,t) = - \frac{1}{r \lambda(r,t)} \int_0^r \frac{\partial \lambda(\rho, t)}{\partial t} \rho \,\mathrm{d}\rho.	\label{eq:v_solution_integral_form}
	\end{equation}
\end{Corollary}

\begin{proof}
	In polar coordinates under radial symmetry, the conservative PDE \eqref{eq:conservative_pde} simplifies to:
	\begin{equation} % eq.24
		\frac{1}{r}\frac{\partial}{\partial r}(r \lambda v_r) = -\frac{\partial \lambda}{\partial t}.
	\end{equation}
	Multiplying by $r$ and integrating from $0$ to $r$:
	\begin{equation}\label{eqCor1} % eq.25
		\int_0^r \frac{\partial}{\partial \rho} (\rho \lambda v_\rho) \,\mathrm{d}\rho = - \int_0^r \rho \frac{\partial \lambda}{\partial t} \,\mathrm{d}\rho.
	\end{equation}
	Applying the Fundamental Theorem of Calculus to the left hand side (LHS) of (\ref{eqCor1}) subject to the non-singular boundary condition $\lim_{r\to 0} r\lambda v_r = 0$:
	\begin{equation} % eq.26
		r \lambda(r,t) v_r(r,t) = - \int_0^r \frac{\partial \lambda(\rho, t)}{\partial t} \rho \,\mathrm{d}\rho.
	\end{equation}
	Dividing the above result by $r\lambda(r,t)$ yields (\ref{eq:v_solution_integral_form}).
\end{proof}

Theorem~\ref{Theorem1} guarantees that the macroscopic velocity field is physically well-defined for any network topology. Corollary~\ref{Cor_Symmetric} further provides the computational basis for radially symmetric systems. Intuitively, \eqref{eq:v_solution_integral_form} states that the outward velocity at radius $r$ is directly proportional to the rate of mass accumulation within the enclosed volume, normalized by the local boundary density. This explicitly links the kinematic ``Digital Tide'' to the conservation of information demand.

%=======================================================================================

\subsection{Analytical Properties and Physical Interpretation}\label{sec:field_properties} % S3.4

The macroscopic velocity field $\bm{v}(\bm{x},t)$, established in Theorem~\ref{Theorem1} as the unique minimum-energy solution to the inverse boundary value problem, constitutes the fundamental kinematic object within the F-STSG framework. Unlike heuristic mobility vectors often employed in simulation-based studies, this derived field possesses specific analytical properties that facilitate rigorous theoretical development:

\begin{enumerate}
	\item \textit{Analytical Tractability:} As demonstrated by Corollary~\ref{Cor_Symmetric}, the field admits closed-form integral representations for canonical topologies. This allows the macroscopic transport dynamics to be computed directly from the parameters of the intensity field $\lambda(\bm{x},t)$, avoiding the computational complexity of numerical fluid dynamics solvers.
	\item \textit{Physical Causality:} The solution is derived strictly from the conservation of mass. It establishes a deterministic link between the temporal rate of change of the node intensity, $\partial_t \lambda$, and the spatial divergence of the transport flux, $\nabla \cdot (\lambda \bm{v})$. This mathematically couples the temporal evolution of demand to the spatial advection of the network topology.
	\item \textit{Statistical Consistency:} Since $\bm{v}(\bm{x},t)$ is a functional of the intensity $\lambda(\bm{x},t)$, it ensures that the Lagrangian transport of the fluid is asymptotically consistent with the Eulerian evolution of the underlying point process statistics. This guarantees that the flux-based metrics derived in Section~\ref{sec:fluid_metrics} are mathematically compatible with the stochastic interference analysis.
\end{enumerate}

This vector field $\bm{v}(\bm{x},t)$ serves as the basis for the perturbation analysis in Section~\ref{sec:linear_response} and the rigorous definition of the Information Flux in Section~\ref{sec:metrics_toolbox}.

%======================================================================
%=====================Section IV=======================================
%======================================================================
\vspace*{-2mm}
\section{Perturbation Analysis and Green's Function Formalism}\label{sec:linear_response} % S4

The governing field equation derived in Section~\ref{sec:core_theory}, $\partial_t \lambda = \nabla \cdot (\lambda \nabla \phi)$, constitutes a linear elliptic PDE for $\phi$ with spatially varying coefficients. While Theorem~\ref{Theorem1} guarantees solution existence, obtaining closed-form expressions for general topologies is analytically intractable due to the spatial inhomogeneity of the divergence operator $\nabla \cdot (\lambda \nabla)$.
However, many physical regimes of interest, such as the superposition of local demand fluctuations upon a smooth macroscopic background, can be modeled using perturbation theory.
In this section, we develop a Green's function framework to derive analytical solutions for arbitrary fields in the weak-variation limit.

\vspace*{-2mm}
\subsection{Linearized Field Equation}\label{sec:linearization} % S4.1

We decompose the general intensity field $\lambda(\bm{x},t)$ into a known background state $\lambda_0(\bm{x},t)$ and a zero-mean perturbation $\delta\lambda(\bm{x},t)$ (i.e., $\int \delta\lambda \,\mathrm{d}\bm{x} = 0$):
\begin{align} % eqs.27,28
	\lambda(\bm{x},t) &= \lambda_0(\bm{x},t) + \epsilon \delta\lambda(\bm{x},t), \label{eq:perturb_lambda} \\
	\phi(\bm{x},t) &= \phi_0(\bm{x},t) + \epsilon \delta\phi(\bm{x},t), \label{eq:perturb_phi}
\end{align}
where $\epsilon \ll 1$ is the perturbation parameter. The background pair $(\lambda_0, \phi_0)$ is assumed to satisfy the zeroth-order conservation equation identically:
\begin{equation} % eq.29
	\frac{\partial \lambda_0}{\partial t} = \nabla \cdot (\lambda_0 \nabla \phi_0).
	\label{eq:zeroth_order}
\end{equation}

\begin{Theorem}\label{Theorem2} %[Linearized Flow Potential Equation] % The.2
	Linearizing the system with respect to $\epsilon$, the potential perturbation $\delta\phi(\bm{x},t)$ is governed by the following linear inhomogeneous PDE:
	\begin{equation} % eq.30
		\nabla \cdot (\lambda_0 \nabla (\delta\phi)) = \underbrace{\frac{\partial (\delta\lambda)}{\partial t} - \nabla \cdot (\delta\lambda \nabla \phi_0)}_{\text{Effective Source Term } S_{\mathrm{eff}}(\bm{x},t)}.
		\label{eq:linearized_pde}
	\end{equation}
	Identifying the background velocity field as $\bm{v}_0(\bm{x},t) \triangleq -\nabla \phi_0(\bm{x},t)$, the source term simplifies to $S_{\mathrm{eff}} = \partial_t (\delta\lambda) + \nabla \cdot (\delta\lambda \bm{v}_0)$.
	The term on the LHS, specifically $\nabla \cdot (\lambda_0 \nabla (\delta\phi))$, represents the diffusion of the potential perturbation within the background medium $\lambda_0$, driven by the effective source $S_{\mathrm{eff}}$ which captures the net mass imbalance created by the density fluctuation.
\end{Theorem}

\begin{proof}
	See Appendix~\ref{app:proof_theorem2}.
\end{proof}

Equation \eqref{eq:linearized_pde} is a linear PDE for the unknown perturbation $\delta\phi$. Crucially, the operator $\mathcal{L}_0 = \nabla \cdot (\lambda_0 \nabla)$ depends only on the known background state, rendering the problem solvable via the Green's function method.

\vspace*{-2mm}
\subsection{Green's Function Integral Representation}\label{sec:greens_function} % S4.2

The linearity of the perturbation equation \eqref{eq:linearized_pde} allows us to formulate the general solution using the Green's function method. Let $G(\bm{x}, \bm{y}; t)$ be the generalized Green's function for the background operator $\mathcal{L}_0 = \nabla \cdot (\lambda_0 \nabla)$ at time $t$.
Since the problem is subject to Neumann boundary conditions (no-flux), the standard definition $\mathcal{L}_0 G = \delta$ is ill-posed due to the global mass conservation constraint. Instead, $G$ is defined as the unique solution (orthogonal to constants) to:
\begin{equation} % eq.31
	\nabla_{\bm{x}} \cdot (\lambda_0(\bm{x},t) \nabla_{\bm{x}} G(\bm{x}, \bm{y}; t)) = \delta(\bm{x} - \bm{y}) - \frac{1}{|\Omega|},
\end{equation}
subject to $\frac{\partial G}{\partial n} = 0$ on $\partial \Omega$.
Because the effective source term $S_{\mathrm{eff}}$ represents a mass redistribution with zero net integral, i.e., $\int_{\Omega} S_{\mathrm{eff}} \,\mathrm{d}\bm{x} = 0$, the constant offset term $-1/|\Omega|$ vanishes upon integration.
The potential perturbation generated by any arbitrary density fluctuation is thus rigorously given by the convolution integral:
\begin{equation} % eq.32
	\delta\phi(\bm{x},t) = \int_{\Omega} G(\bm{x}, \bm{y}; t) S_{\mathrm{eff}}(\bm{y},t) \,\mathrm{d}\bm{y}.
	\label{eq:greens_solution}
\end{equation}
Consequently, the velocity perturbation field can be expressed explicitly as:
\begin{equation} % eq.33
	\delta \bm{v}(\bm{x},t) = - \nabla_{\bm{x}}\! \left( \int_{\Omega} G(\bm{x}, \bm{y}; t) \left( \frac{\partial \delta\lambda}{\partial t} + \nabla_{\bm{y}} \cdot (\delta\lambda \bm{v}_0) \right) \mathrm{d}\bm{y}\! \right)\! .
\end{equation}
This formalism establishes that the F-STSG framework extends beyond symmetric geometries. Through \eqref{eq:greens_solution}, the dynamics of complex, non-symmetric topologies can be analyzed by decomposing the field into a superposition of eigenmodes. The Green's function $G$ acts here as the fundamental propagator, diffusing local density fluctuations across the network domain.

\subsection{Spectral Decomposition of Network Propagator}\label{sec:spectral_properties} % S4.3

While the integral representation in \eqref{eq:greens_solution} provides the general solution, the structure of the perturbation dynamics is most rigorously understood through its spectral properties.
The differential operator governing the background flow, $\mathcal{L}_0 \phi = \nabla \cdot (\lambda_0 \nabla \phi)$, acts as a generalized Laplacian with spatially varying conductivity.
Crucially, under the Neumann boundary conditions derived in Theorem~\ref{Theorem1}, this operator is self-adjoint with respect to the standard $L^2$ inner product.
According to standard Sturm-Liouville theory for elliptic operators on bounded domains~\cite{Courant1953}, there exists a countable set of orthonormal eigenfunctions $\{\psi_n(\bm{x})\}$ and corresponding non-negative eigenvalues, $0 = \nu_0 < \nu_1 \le \nu_2 \le \dots$, satisfying:
\begin{equation} % eq.34
	\nabla \cdot (\lambda_0(\bm{x}) \nabla \psi_n(\bm{x})) = -\nu_n \psi_n(\bm{x}),
\end{equation}
subject to the Neumann boundary condition $\partial \psi_n / \partial n = 0$ on $\partial \Omega$.
Note that the zero eigenvalue $\nu_0=0$ corresponds to the constant eigenfunction $\psi_0 = |\Omega|^{-1/2}$, representing the conservation of total mass.

The generalized Green's function $G(\bm{x}, \bm{y})$, defined on the subspace orthogonal to $\psi_0$, admits the following eigenfunction expansion (Mercer's theorem):
\begin{equation} % eq.35
	G(\bm{x}, \bm{y}) = - \sum_{n=1}^{\infty} \frac{\psi_n(\bm{x}) \psi_n(\bm{y})}{\nu_n}.
	\label{eq:spectral_expansion}
\end{equation}
Substituting this spectral form into the perturbation solution \eqref{eq:greens_solution} yields:
\begin{equation} % eq.36
	\delta\phi(\bm{x},t) = - \sum_{n=1}^{\infty} \frac{\psi_n(\bm{x})}{\nu_n} \underbrace{\left( \int_{\Omega} \psi_n(\bm{y}) S_{\mathrm{eff}}(\bm{y},t) \,\mathrm{d}\bm{y} \right)}_{\text{Modal Projection } \hat{S}_n(t)}.
\end{equation}

Physically, \eqref{eq:spectral_expansion} reveals the spatial spectral response of the network topology. The eigenfunctions $\psi_n(\bm{x})$ represent the natural spatial eigenmodes of traffic flow adjustments, while the eigenvalues $\nu_n$ quantify the topological resistance to the $n$-th mode.
Low-order modes characterized by small $\nu_n$ correspond to long-range, global transport which decays slowly (as $1/\nu_n$ is large), whereas high-order modes associated with large $\nu_n$ represent localized fluctuations that are rapidly attenuated.
This spectral hierarchy implies that the network acts as a low-pass spatial filter, suppressing high-frequency stochastic noise while effectively propagating macroscopic demand trends.

%======================================================================
%=====================Section V ============================
%======================================================================
\section{Field-Theoretic Metrics and Dynamic Performance Analysis}\label{sec:metrics_toolbox} % S5

The preceding sections established the core theoretical framework: the reconstruction of the macroscopic velocity field $\bm{v}(\bm{x},t)$ as the unique minimum-energy solution to the inverse boundary value problem.
We now leverage this derived field to characterize the system dynamics.
Unlike heuristic indicators used in traditional network monitoring, the metrics derived herein are process-endogenous, stemming directly from the conservation laws governing the hydrodynamic limit.
This section formulates two complementary classes of metrics: Lagrangian kinematics, describing the intrinsic transport of the information demand, and dynamic network performance, quantifying the coupling between this transport and the stationary infrastructure.

\subsection{Lagrangian Kinematics of Information}\label{sec:fluid_metrics} % S5.1

These metrics quantify the collective mobility of the node constellation, mathematically independent of the specific network infrastructure. They serve as the fundamental predictors of topological evolution.

\subsubsection{The Information Flux Vector}\label{sec:metric_flux} % S5.1.1)
Generalizing the concept of microscopic node kinetics, we define the Information Flux, a vector field quantifying the spatial advection of the intensity field.

\begin{Definition}\label{Definition4} %[Information Flux] % Def.4
	The Information Flux, $\bm{J}(\bm{x},t)$, is defined as the product of the local intensity and the macroscopic velocity field:
	\begin{equation} % eq.37
		\bm{J}(\bm{x},t) \triangleq \lambda(\bm{x},t) \bm{v}(\bm{x},t).
	\end{equation}
	$\bm{J}(\bm{x},t)$ serves as the kinematic generator of the ASE.
	Specifically, let $\bar{R}(\bm{x},t) \triangleq \mathbb{E}[\log(1+\text{SINR})]$ denote the local ergodic SE per node.
	In the interference-limited cellular regime, $\bar{R}$ exhibits signal-to-interference ratio (SIR) scale invariance. As established in SG \cite{Andrews2011}, densification scales both signal and interference power proportionally. Under the assumption of local homogeneity, the ergodic efficiency satisfies $\bar{R} \approx \bar{R}_{\infty}$, asymptotically independent of the spatial density $\lambda$.
	Consequently, the macroscopic transport of Shannon capacity, denoted as the Capacity Flux $\bm{J}_{\mathrm{cap}}$, is dominated by the density dynamics: $\bm{J}_{\mathrm{cap}} \approx \bar{R}_{\infty} \bm{J}(\bm{x},t)$.
	Thus, a non-zero flux vector $\bm{J}$ directly quantifies the spatial advection of the network capacity region~\cite{ElGamal2011}, necessitating the dynamic redistribution of spectral resources.
\end{Definition}

This definition allows us to quantify the precise flow of demand across any arbitrary geometric boundary $\mathcal{B}$, such as a cell edge or a tracking area border.

\begin{remark}%[Equivalence to Capacity Flux] % Rem.4
	It is pertinent to verify whether $\bm{J}$, defined via node density, truly represents information flow.
	In the hydrodynamic limit of interference-limited networks, the local ASE scales linearly with node density: $\mathcal{T}(\bm{x},t) \approx \lambda(\bm{x},t) \cdot \bar{R}_{\infty}$, as analytically derived in Section~\ref{sec:app_optimization}.
	Applying the continuity equation, the temporal evolution of capacity is governed by:
	\begin{equation} % eq.38
		\frac{\partial \mathcal{T}}{\partial t} \approx \bar{R}_{\infty} \frac{\partial \lambda}{\partial t} = - \nabla \cdot (\bar{R}_{\infty} \bm{J}).
	\end{equation}
	This identity proves that $\bm{J}(\bm{x},t)$ is indeed the sufficient statistic governing the dynamic reshaping of the network's spatial capacity profile.
	
	\textit{Dimensional Analysis:} If $\lambda$ has dimension $[L^{-d}]$ (nodes/volume) and $\bm{v}$ has dimension $[L T^{-1}]$, then $\bm{J}$ has dimension $[L^{-(d-1)} T^{-1}]$. Scaling by the SE $\bar{R}_{\infty}$ [bits/s/Hz/node], the capacity flux $\bm{J}_{\mathrm{cap}}$ carries the dimension of [bits/s/Hz/m$^{d-1}$], rigorously representing the flow of spectral capacity across a geometric cross-section.
\end{remark}

\begin{Proposition}\label{Prop_Flux_Integral} %[Net Flux Across a Manifold] % Pro.2
	Let $\mathcal{B}$ be a $(d-1)$-dimensional manifold in $\mathbb{R}^d$ with outward normal $\mathbf{n}$.
	The net rate of information demand crossing $\mathcal{B}$ at time $t$, denoted $J_{\mathcal{B}}(t)$, is given by the surface integral:
	\begin{equation} % eq.39
		J_{\mathcal{B}}(t) = \int_{\mathcal{B}} \bm{J}(\bm{x},t) \cdot \mathbf{n} \,\mathrm{d}S.
		\label{eq:handover_flux}
	\end{equation}
\end{Proposition}

\begin{proof}
	See Appendix~\ref{app:proof_theorem3}.
\end{proof}

Unlike reactive counting measures that register transport events post-facto, the flux integral $J_{\mathcal{B}}(t)$ serves as a predictive kinematic statistic. It quantifies the instantaneous momentum of demand directed at a boundary, providing a sufficient statistic for anticipating future load imbalances before they manifest as local capacity outages.

\subsubsection{Topological Compression and Material Derivative}\label{sec:metric_divergence} % S5.1.2)
To rigorously characterize the formation of demand hotspots, we analyze the compressibility of the flow.

\begin{Definition}\label{Definition6} %[Congestion Divergence] % Def.5
	The Congestion Divergence, $\mathcal{D}(\bm{x},t)$, is the scalar field defined by the divergence of the macroscopic velocity:
	\begin{equation} % eq.40
		\mathcal{D}(\bm{x},t) \triangleq \nabla \cdot \bm{v}(\bm{x},t).
	\end{equation}
\end{Definition}

\begin{Proposition}\label{Prop_Material_Derivative} %[Material Derivative as a Stability Indicator] % Pro.3
	For a conservative system ($S=0$), the Congestion Divergence is mathematically equivalent to the negative logarithmic rate of change of the intensity, as observed in the Lagrangian frame of reference:
	\begin{equation}\label{eqPro3} % eq.41
		\mathcal{D}(\bm{x},t) = - \frac{1}{\lambda(\bm{x},t)} \frac{D\lambda(\bm{x},t)}{Dt},
	\end{equation}
	where $\frac{D}{Dt} \triangleq \frac{\partial}{\partial t} + \bm{v} \cdot \nabla$ is the material (or Lagrangian) derivative, capturing the rate of change experienced by a moving fluid element.
\end{Proposition}

\begin{proof}
	See Appendix~\ref{app:proof_theorem4}.
\end{proof}

Proposition~\ref{Prop_Material_Derivative} establishes the material derivative $\frac{D\lambda}{Dt}$ as the fundamental indicator of network stability. A region characterized by negative divergence $\mathcal{D} < 0$ acts as a topological sink, where the flow field physically compresses the node constellation, inevitably leading to a density surge. This offers a kinematic guarantee of future congestion, providing an early-warning signal distinct from simple static density thresholds.

\subsubsection{Global Centroid Dynamics}\label{sec:metric_centroid} % S5.1.3)
To track the global trajectory of the demand, we define the first moment of the intensity field.

\begin{Definition}\label{Definition5} %[Network Centroid] % Def.6
	The Network Centroid, $\bm{C}(t)$, is the instantaneous expectation of the node location vector:
	\begin{equation} % eq.42
		\bm{C}(t) \triangleq \mathbb{E}_{\bm{x} \sim p(\cdot,t)}[\bm{x}] = \frac{1}{N(t)} \int_{\mathbb{R}^d} \bm{x} \lambda(\bm{x},t) \,\mathrm{d}\bm{x}.
	\end{equation}
\end{Definition}

\begin{Proposition}\label{Prop_Centroid_Drift} %[Centroid Drift Velocity]
	In a conservative system ($\frac{\mathrm{d}N}{\mathrm{d}t} = 0$), the velocity of the centroid, $\bm{V}_C(t) \triangleq \frac{\mathrm{d}\bm{C}}{\mathrm{d}t}$, is exactly the intensity-weighted spatial average of the macroscopic velocity field:
	\begin{equation} % eq.43
		\bm{V}_C(t) = \mathbb{E}_{\bm{x} \sim p(\cdot,t)}[\bm{v}(\bm{x},t)] = \int_{\mathbb{R}^d} \bm{v}(\bm{x},t) p(\bm{x},t) \,\mathrm{d}\bm{x}.
	\end{equation}
\end{Proposition}
\begin{proof}
	See Appendix~\ref{app:proof_theorem5}.
\end{proof}

\subsection{Dynamic Performance under Field-Measure Coupling}\label{sec:network_metrics} % S5.2

We now map the continuous kinematics derived above onto the discrete performance of the infrastructure. To justify the calculation of instantaneous spatial expectations on a non-stationary field, we formally characterize the system's operating regime via the adiabatic approximation.

\begin{Assumption}\label{Assump:Adiabatic}%[The Adiabatic Regime] % Assu.1
	Let $T_{\mathrm{packet}}$ be the packet transmission duration, and let $T_{\mathrm{hydro}}$ be the characteristic time scale of the macroscopic topological evolution, defined as the inverse of the maximum normalized material derivative:
	\begin{equation} % eq.44
		T_{\mathrm{hydro}} \triangleq \left( \sup_{\bm{x},t} \left| \frac{1}{\lambda} \frac{D\lambda}{Dt} \right| + \frac{\sup \|\bm{v}\|}{L_{\mathrm{cell}}} \right)^{-1}.
	\end{equation}
	where $L_{\mathrm{cell}}$ is the characteristic inter-site distance. This composite metric captures both topological compression and cell-crossing advection dynamics.
	We define the non-stationarity parameter as the ratio $\epsilon \triangleq T_{\mathrm{packet}} / T_{\mathrm{hydro}}$.
	The F-STSG framework operates in the adiabatic (or quasi-static) limit where $\epsilon \to 0$, under the condition that the channel coherence time satisfies $T_{\mathrm{coh}} \ge T_{\mathrm{packet}}$.
	In this regime, the intensity field $\lambda(\bm{x},t)$ is treated as quasi-static during a transmission event, allowing instantaneous performance to be computed as a snapshot expectation over the frozen Palm measure.
\end{Assumption}

{\color{black}\begin{remark}%[Heuristic Error Scaling and Validity] % Rem.5
	Assumption~\ref{Assump:Adiabatic} defines the asymptotic validity of the theory. For systems with finite time-scale separation ($\epsilon > 0$), we heuristically approximate the error in the calculated coverage probability to scale as $\mathcal{O}(\epsilon)$. This first-order scaling is postulated under the assumption that the performance functional $P_{\mathrm{cov}}(\lambda)$ is locally Lipschitz continuous with respect to the intensity field. During a transmission event $T_{\mathrm{packet}}$, the density perturbation is linearly bounded by $\Delta \lambda \approx (\partial_t \lambda) T_{\mathrm{packet}} \propto \epsilon$. Physically, this error reflects the unmodeled Lagrangian displacement of mobile nodes during a single packet transmission. For typical urban mobility ($v \approx 10$ m/s) and cellular frame structures ($T_{\mathrm{packet}} \approx 1$ ms), we have $\epsilon \approx 10^{-4}$, rendering the zero-order adiabatic assumption practically robust, thereby circumventing the need for a rigorous functional perturbation proof.
\end{remark}}

\begin{remark}%[Link to Non-Ergodic Channel Capacity] % Rem.6
	From an information-theoretic perspective, the time-varying topology $\lambda(\bm{x},t)$ induces a non-ergodic macroscopic channel. Classical ergodic capacity is defined over infinite time. However, due to the macroscopic evolution $T_{\mathrm{hydro}}$, the system effectively operates as a block-fading channel where the block is the quasi-static epoch defined by the macroscopic flow. Therefore, our metric $P_{\mathrm{cov}}(t)$ is rigorously equivalent to the complement of the information outage probability conditioned on the instantaneous macroscopic state $\lambda_t$.
\end{remark}

\subsubsection{Instantaneous Association Probability} % S5.2.1)
The association preference is governed by the rule in \eqref{eq:assoc_rule}. The macroscopic association behavior is characterized by the spatial expectation of local preferences over the dynamic demand field.

\begin{Proposition}\label{Prop_Assoc_Prob} %[Dynamic Association Probability] % Pro.5
	The probability that a typical mobile node associates with tier $k$ at time $t$, denoted $A_k(t)$, is derived by weighting the location-specific association probability $\mathcal{A}_k(\bm{x},t)$ by the instantaneous normalized intensity field:
	\begin{equation} % eq.45
		A_k(t) = \mathbb{E}_{\bm{x} \sim p(\cdot,t)} \left[ \mathcal{A}_k(\bm{x},t) \right] =\!\! \int_{\mathbb{R}^d}\!\! \mathcal{A}_k(\bm{x},t) \frac{\lambda(\bm{x},t)}{N(t)} \,\mathrm{d}\bm{x},\!
		\label{eq:assoc_prob_general}
	\end{equation}
	where $\mathcal{A}_k(\bm{x},t) = \mathbb{P}(\mathcal{K}(\bm{X}^*) = k \mid \bm{x})$ is determined by the SG of the active infrastructure $\Psi(t)$.
\end{Proposition}

\begin{proof}
	See Appendix~\ref{app:proof_theorem6}.
\end{proof}

\subsubsection{Instantaneous Coverage Probability (Information Outage)} % S5.2.2)
The central performance metric is the Instantaneous Coverage Probability, $P_{\mathrm{cov}}(t)$. From an information-theoretic perspective, the condition $\text{SINR} > \gamma_k$ is equivalent to the instantaneous channel capacity $C = \log_2(1+\text{SINR})$ exceeding a target transmission rate $R_k = \log_2(1+\gamma_k)$. Thus, $P_{\mathrm{cov}}(t)$ is defined as the complement of the Information Outage Probability for a randomly selected node within the dynamic field.

\begin{Proposition}\label{Prop_Cov_Prob} %[Dynamic Coverage Probability] % Pro.6
	The expected instantaneous network coverage probability at time $t$, averaged over the spatial statistics of the infrastructure processes, is given by:
	\begin{equation} % eq.46
		P_{\mathrm{cov}}(t) = \sum_{k=1}^K A_k(t) \cdot P_{\mathrm{c}}^{(k)}(t),
		\label{eq:pcov_instantaneous}
	\end{equation}
	where $P_{\mathrm{c}}^{(k)}(t)$ is the conditional probability that a node achieves its SINR target $\gamma_k$, given association with tier $k$:
	\begin{equation} % eq.47
		P_{\mathrm{c}}^{(k)}(t) \triangleq \int_{\mathbb{R}^d} \mathbb{P}(\text{SINR} > \gamma_k \mid \bm{x}, k) \, p_k(\bm{x},t) \,\mathrm{d}\bm{x}.
	\end{equation}
	Here, $p_k(\bm{x},t)$ is the probability density of user locations conditioned on association with tier $k$, given explicitly by:
	\begin{equation} % eq.48
		p_k(\bm{x},t) = \frac{\mathcal{A}_k(\bm{x},t)\lambda(\bm{x},t)}{N(t) A_k(t)}.
	\end{equation}
\end{Proposition}

\begin{proof}
	See Appendix~\ref{app:proof_theorem7}.
\end{proof}

Equation~\eqref{eq:pcov_instantaneous} mathematically encapsulates the field-measure coupling of the F-STSG framework. The term $A_k(t)$ is driven by the Lagrangian transport of the intensity field representing node mobility, while $P_{\textrm{c}}^{(k)}(t)$ is determined by the Eulerian configuration of the infrastructure governing the interference field.
	
%======================================================================
%=====================Section VI=======================================
%======================================================================

\vspace*{-1mm}
\section{Simulation and Validation}\label{sec:simulation} % S6

	We validate the F-STSG framework by confronting macroscopic analytical predictions against microscopic event-driven Monte Carlo ground truth. The validation is rigorously structured into three logical tiers:
	\begin{itemize}
		\item \textbf{Thermodynamic Consistency (Section~\ref{sec:sim_consistency}):} We verify the convergence of the discrete empirical measure to the hydrodynamic limit, confirming the $\mathcal{O}(N^{-1/2})$ scaling law.
		{\color{black}\item \textbf{Numerical Consistency and Solver Robustness (Sections~\ref{sec:sim_velocity} and \ref{sec:sim_topology}):} We assess the numerical consistency of the discrete operator, its capability to filter topologically stationary vorticity, and its stability under complex non-symmetric geometries.}
		\item \textbf{Dynamic Metrics and Performance (Sections~\ref{sec:sim_mobility_metrics} and \ref{sec:sim_performance}):} We validate the Lagrangian kinematic metrics (Information Flux) and confirm the end-to-end adiabatic coupling between the continuous intensity field and discrete network coverage probability.
\end{itemize}

%===================new=====================
\vspace*{-1mm}
\subsection{Asymptotic Consistency and Thermodynamic Limit}\label{sec:sim_consistency} % S6.1

The foundational premise of F-STSG is the hydrodynamic limit established in Section~\ref{sec:model_user_population}, postulating the weak convergence of the empirical measure $\mu_N$ to the continuum field $\lambda$ as $N \to \infty$.
We quantify the validity of this approximation by the asymptotic decay of the relative reconstruction error, defined as:
\begin{equation} % eq.49
	\mathcal{E}(N) \triangleq \frac{\| \hat{\lambda}_N(\bm{x},t) - \lambda(\bm{x},t) \|_{L^2}}{\| \lambda(\bm{x},t) \|_{L^2}},
\end{equation}
where $\lambda(\bm{x},t)$ is the theoretical macroscopic intensity derived from the continuum model, and $\hat{\lambda}_N(\bm{x},t)$ is the empirical density field reconstructed from the discrete Monte Carlo node constellation.

	\begin{figure}[!b]
	\vspace*{-6mm}
	\centering
	\includegraphics[width=0.9\columnwidth]{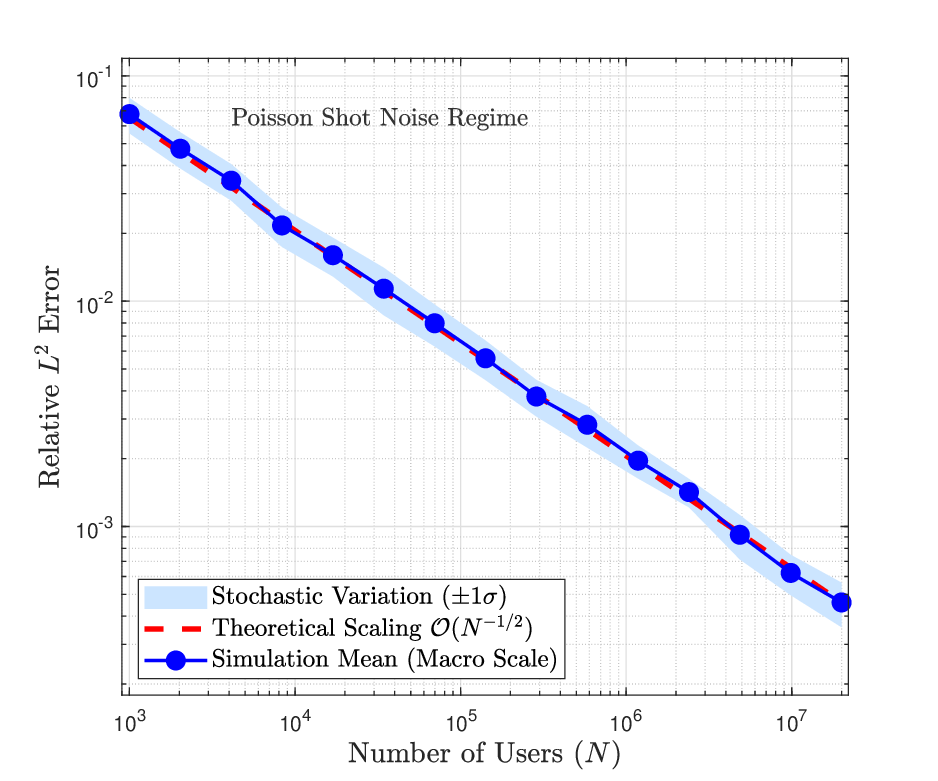} 
	\vspace*{-5mm}
	\caption{Validation of the Thermodynamic Limit: Relative reconstruction error $\mathcal{E}(N)$ vs. population size $N$.}
	\label{fig:convergence} % Fig.2
\vspace*{-1mm}
\end{figure}

\subsubsection{Simulation Setup} % S6.1.1)
We simulate a non-separable dynamic topology characterized by a time-varying Gaussian hotspot. 
To distinguish stochastic convergence from discretization bias, the theoretical ground truth $\lambda(\bm{x},t)$ is integrated using the cumulative distribution function (CDF) difference method. 
In the Monte Carlo phase, for each realization $m\! \in\! \{1, \dots, M\}$, we generate a discrete node constellation $\Phi^{(N)}_m$ by drawing $N$ random coordinate samples according to the instantaneous density profile.
We then reconstruct the empirical intensity $\hat{\lambda}_N$ using a mesoscopic smoothing kernel with a bandwidth sufficiently larger than the typical inter-node distance.
This effectively isolates the bulk transport phenomena, the Digital Tide, from high-frequency microscopic singularities. 
The statistical results presented are averaged over $M\! =\! 10^5$ independent Monte Carlo realizations for each population size $N$.
	
	\subsubsection{Convergence Analysis} % S6.1.2)
	Fig.~\ref{fig:convergence} plots the evolution of $\mathcal{E}(N)$ against the population size $N$. The numerical error (black markers) exhibits a strictly linear decay in the log-log domain with a slope of $-0.5$, strictly adhering to the theoretical Poisson shot noise scaling $\mathcal{O}(N^{-1/2})$. The narrow $\pm 1\sigma$ confidence interval (shaded region) further confirms the measure concentration of the system. As $N$ increases, the variance of individual snapshots vanishes. Crucially, in the ultra-dense regime ($N \approx 10^7$), the error drops below $10^{-3}$, confirming that the Digital Tide approximation becomes asymptotically exact and justifying the use of deterministic PDEs for network state evolution.
	
		\begin{figure}[!b]
		\vspace*{-6mm}
		\centering
		\hspace*{-10mm}\includegraphics[width=1.2\columnwidth]{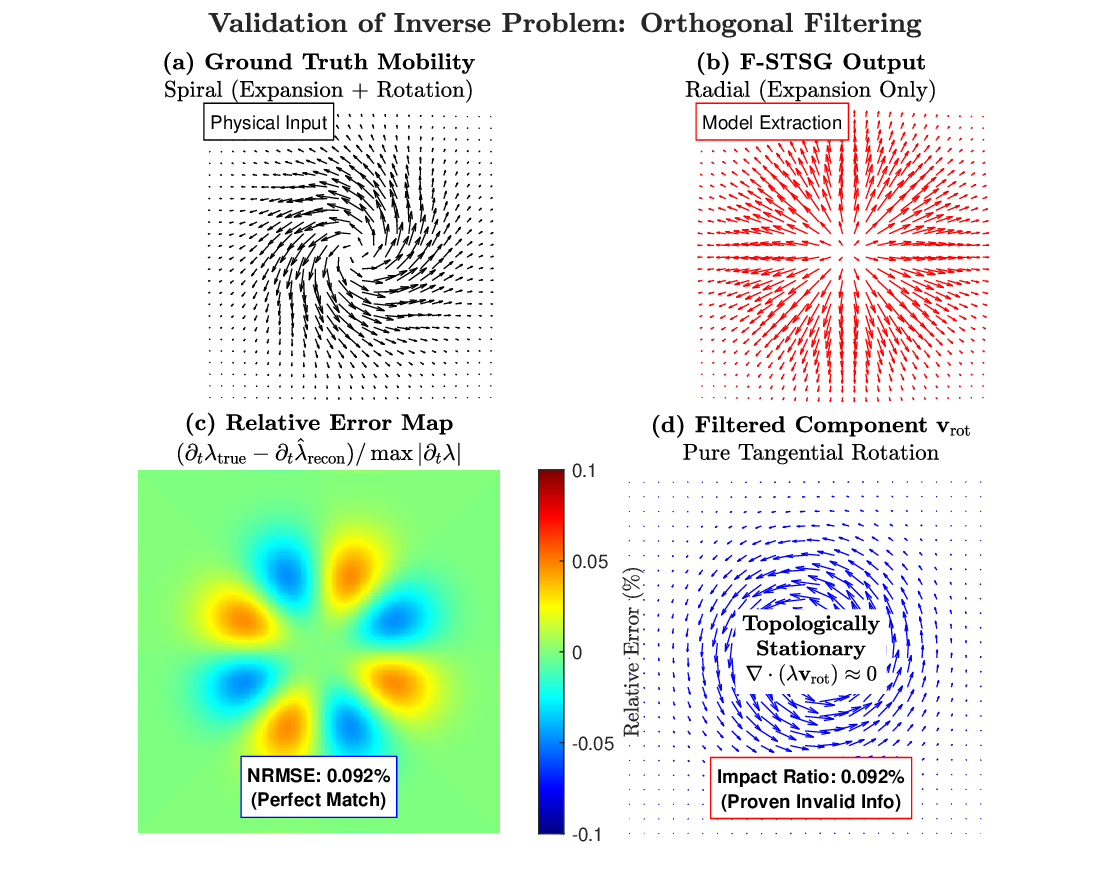}
		\vspace*{-9mm}
		\caption{Validation of the inverse problem. (a) Ground truth velocity $\bm{v}_{\mathrm{true}}$ (superposition of expansion and rotation). (b) Reconstructed field $\hat{\bm{v}}$ (irrotational). (c) Relative prediction error map. (d) Residual component $\bm{v}_{\mathrm{true}} - \hat{\bm{v}}$, corresponding to the filtered vorticity.}
		\label{fig:inverse_validation} % Fig.3
		\vspace*{-1mm}
	\end{figure}
	
	%==================new VI-B======================	
\vspace*{-1mm}
\subsection{Validation of Inverse Problem: Orthogonal Filtering}\label{sec:sim_velocity}% S6.2

Unlike the stochastic convergence analysis in Section~\ref{sec:sim_consistency}, this section validates the mathematical exactness of the inverse solver via deterministic numerical experiments using synthesized continuum fields.
Proposition~\ref{Prop_Optimal_Transport} asserts that the F-STSG framework reconstructs the unique minimum-energy transport component, acting as a spectral filter for kinematic modes that do not contribute to the information flux. 

To verify this property and the resolution of the inverse problem's non-uniqueness, we construct a ground truth velocity field $\bm{v}_{\mathrm{true}}$ as a superposition of a compressive signal and a solenoidal interference component (kinematic noise):
\begin{equation} % eq.50
	\bm{v}_{\mathrm{true}} = \underbrace{\bm{v}_{\mathrm{irr}}}_{\text{Signal}} + \underbrace{\bm{v}_{\mathrm{rot}}}_{\text{Kinematic Noise}}.
\end{equation}
Here, $\bm{v}_{\mathrm{irr}}$ is a radial gradient field driving density evolution. The noise component $\bm{v}_{\mathrm{rot}}$ represents a ``topologically stationary'' rotation (vorticity) designed to satisfy the strict orthogonality condition $\bm{v}_{\mathrm{rot}} \perp \nabla \lambda$, or equivalently $\nabla \cdot (\lambda \bm{v}_{\mathrm{rot}}) \equiv 0$. This component consumes kinetic energy but induces no net mass transport.
	
The F-STSG solver is driven solely by the scalar evolution $\partial_t \lambda$ generated by $\bm{v}_{\mathrm{true}}$. The reconstruction results, summarized in Fig.~\ref{fig:inverse_validation}, empirically confirm the orthogonality of the decomposition.
While the ground truth field in Panel (a) exhibits pronounced spiraling vorticity, Panel (b) demonstrates that the algorithm successfully isolates the pure radial expansion, projecting the dynamics onto the canonical irrotational manifold defined by $\hat{\bm{v}} = -\nabla \phi$.
The fidelity of this projection is quantified in Panel (c): the local relative error remains strictly bounded below $0.1\%$, with a global normalized root mean square error (NRMSE) of approximately $0.092\%$. We note that in this continuum setting, the NRMSE is mathematically equivalent to the relative $L^2$ error defined as $\|\hat{\bm{v}} - \bm{v}_{\mathrm{true}}\|_{L^2} / \|\bm{v}_{\mathrm{true}}\|_{L^2}$. This residual error is attributable to high-order discretization artifacts rather than model bias.
Finally, Panel (d) visualizes the discarded residual $\bm{v}_{\mathrm{true}} - \hat{\bm{v}}$. This field corresponds precisely to the injected solenoidal interference. The vanishing divergence of this residual confirms that the framework correctly identifies and filters out energetic inefficiencies that are irrelevant to capacity dimensioning.

	{\color{black}\subsection{Numerical Consistency and Spectral Fidelity in Complex Topologies}\label{sec:sim_topology} % S.6.3
	
	Following the kinematic filtration analysis, this section verify the numerical self-consistency and spectral fidelity of the discrete F-STSG operator under complex topological conditions. We specifically examine the reconstruction performance in regimes exhibiting non-convex geometries and saddle-point singularities, where standard gradient-based solvers are prone to numerical divergence.}
	
		\begin{figure}[!b]
		\vspace*{-5mm}
		\centering
		\hspace*{-5mm}\includegraphics[width=0.5\textwidth]{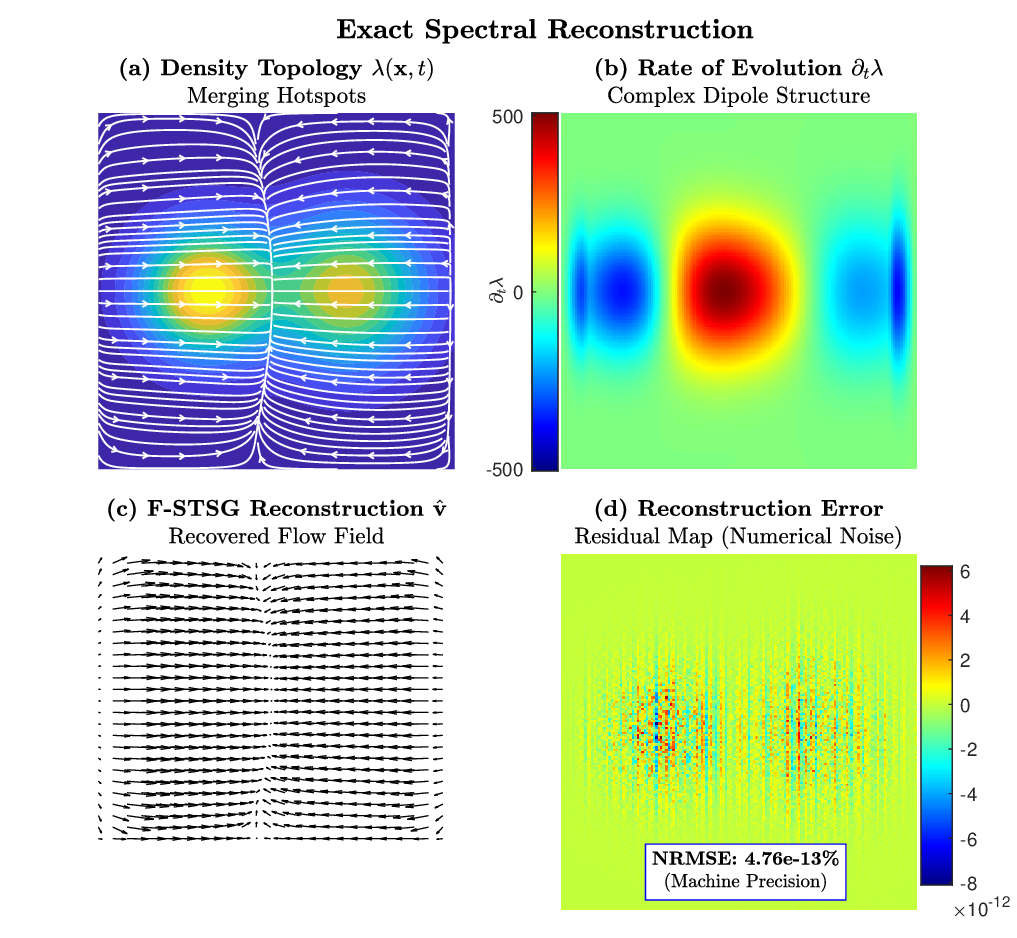}
		\vspace*{-6mm}
		\caption{Validation under complex non-symmetric topology (Hotspot Merging). 
			(a) Non-convex density field $\lambda$ with a saddle point. 
			(b) Scalar evolution rate $\partial_t \lambda$. 
			(c) Reconstructed velocity field $\hat{\bm{v}}$. 
			(d) Relative error map (NRMSE $\approx 10^{-13}\%$).}
		\label{fig:complex_topology} % Fig.4
		\vspace*{-1mm}
	\end{figure}
	
	\subsubsection{Experimental Setup and Consistent Discretization} % S6.3.1)
	We synthesize a Merging Hotspots scenario illustrated in Fig.~\ref{fig:complex_topology}(a), where two asymmetric Gaussian density peaks converge. This configuration generates a topological saddle point at the origin, characterized by a vanishing gradient, i.e., $\nabla \lambda \to \mathbf{0}$, but a non-zero Hessian, serving as a critical assessment of flow resolution in stagnation regions.
	To verify the solver at machine precision and eliminate discretization bias between the forward (ground truth generation) and inverse (spectral reconstruction) operators, we employ a spectrally consistent discretization scheme. Let $\mathcal{D}_h$ denote the central-difference divergence operator on the spatial grid. We construct the inverse spectral operator using the fast Fourier transform with modified wavenumbers $\tilde{k}$, defined as:
	\begin{equation} % eq.51
		\tilde{k}_x = \frac{\sin(k_x \Delta x)}{\Delta x}, \quad \tilde{k}_y = \frac{\sin(k_y \Delta y)}{\Delta y},
	\end{equation}
	where $k_x$ and $k_y$ are the standard wavenumbers, and $\Delta x$ and $\Delta y$ are the grid spacings. This modification ensures that the spectral derivative operator is algebraically equivalent to the spatial central difference operator. Consequently, the ground truth velocity field $\bm{v}_{\mathrm{true}}$ is generated by projecting the raw kinematic superposition onto the curl-free manifold using this consistent spectral basis:
	\begin{equation} % eq.52
		\bm{v}_{\mathrm{true}} = \mathcal{P}_{\text{consistent}}\left( \frac{\lambda_A \bm{v}_A + \lambda_B \bm{v}_B}{\lambda_A + \lambda_B} \right).
	\end{equation}
	This formulation strictly enforces the irrotational constraint $\bm{v}_{\mathrm{true}} = -\nabla \phi$ in the discrete domain, allowing for the isolation of topological solving capability from discretization errors.
	
	\subsubsection{Results and Analysis} % S6.3.2)
	The reconstruction results are summarized in Fig.~\ref{fig:complex_topology}. As visualized in Panel (c), the reconstructed field $\hat{\bm{v}}$ correctly resolves the topological saddle point at the origin, accurately capturing the flow stagnation and redirection inherent to the merging process. 
	Crucially, Panel (d) demonstrates that the reconstruction error is limited solely by floating-point arithmetic. {\color{black}The NRMSE is approximately $10^{-13}\%$, which corresponds to machine-level precision. This result serves as a numerical consistency check, confirming the algebraic closure between the forward density evolution and the spectrally consistent inverse operator. It demonstrates that our discretization scheme accurately preserves the irrotational manifold defined in Theorem~\ref{Theorem1}, ensuring that the inversion process introduces no parasitic numerical vorticity or discretization-induced bias, even in the presence of topological saddle points.}

\begin{figure}[!b]
	\vspace*{-6mm}
	\centering
	\includegraphics[width=\columnwidth]{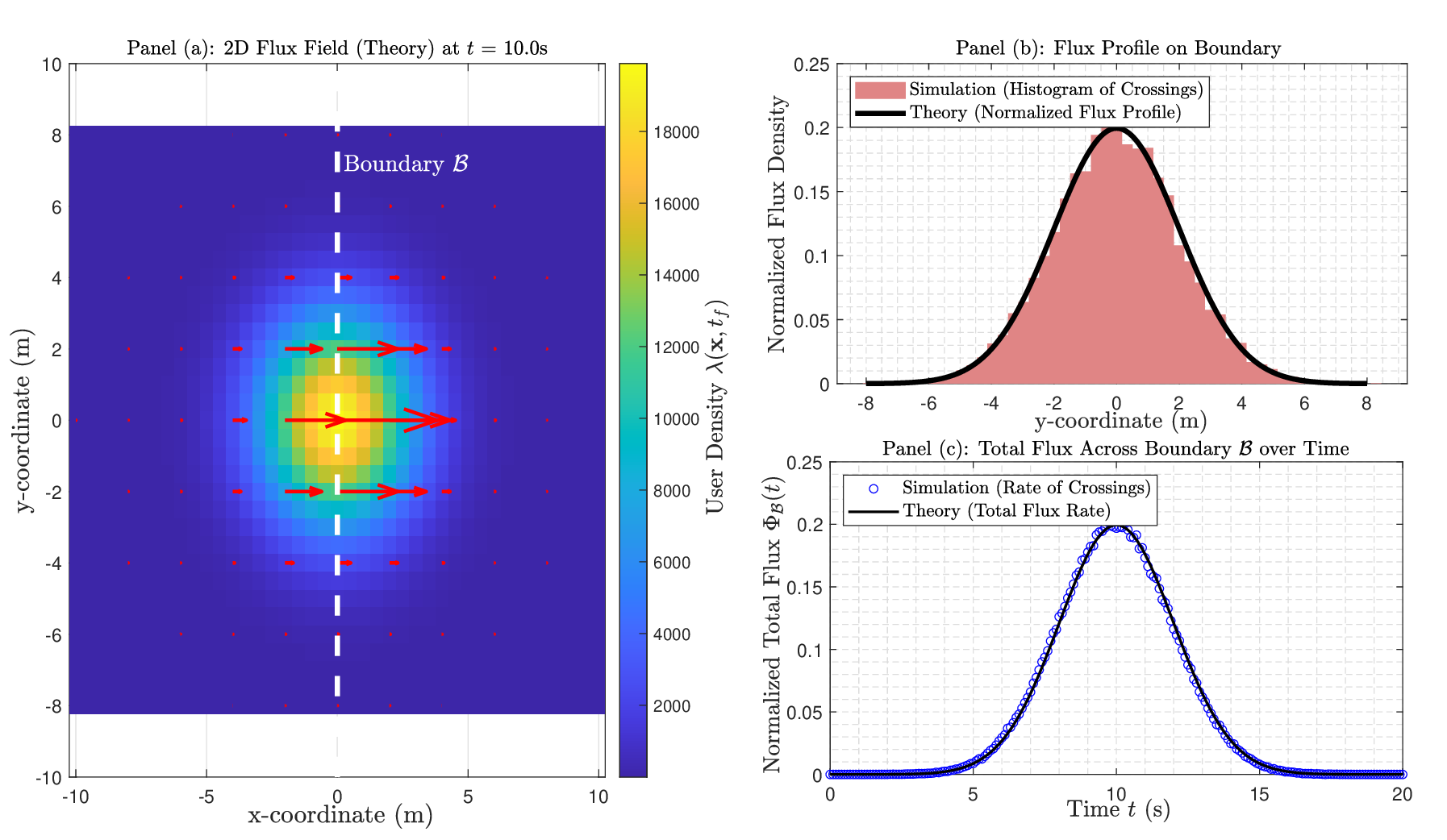} % Use single column width to save space
	\vspace*{-6mm}
	\caption{Validation of Information Flux (Proposition \ref{Prop_Flux_Integral}).}
	\label{fig:handover_flux_validation} % Fig.5
	\vspace*{-1mm}
\end{figure}

\subsection{Validation of Macroscopic Mobility Metrics}\label{sec:sim_mobility_metrics} % S6.4

Building on the kinematic foundation, we validate the derived mobility metrics: Net Flux across Boundaries (Proposition~\ref{Prop_Flux_Integral}), Congestion Divergence (Proposition~\ref{Prop_Material_Derivative}), and Centroid Drift (Proposition~\ref{Prop_Centroid_Drift}).

\subsubsection{Vector Field Consistency (Information Flux)} % S6.4.1)
We scrutinize the Information Flux $\bm{J}(\bm{x},t)$ in a non-symmetric scenario: a Gaussian hotspot translating across a static boundary. Fig.~\ref{fig:handover_flux_validation} presents the multi-dimensional validation. Panel (a) confirms that the analytical flux vectors align with the direction of mass transport. 
Panel (b) validates the spatial consistency of the derived field, showing that the theoretical flux density profile along the boundary perfectly matches the histogram of microscopic crossing locations.
Panel (c) shows that the theoretical flux integral $\int \bm{J} \cdot \mathbf{n} \mathrm{d}S$, depicted as a solid line, accurately predicts the empirical rate of discrete boundary crossings represented by markers. This validates \eqref{eq:handover_flux} as a deterministic predictor of boundary crossing rates.

\subsubsection{Scalar Consistency (Divergence and Centroid)} % S6.4.2)
We further verified the scalar metrics through extensive simulations. The numerical results confirm that the Congestion Divergence $\mathcal{D}(\bm{x},t)$ computed from the velocity field is mathematically identical to the negative material derivative of the density, validating Proposition~\ref{Prop_Material_Derivative}. Note that the precise validation of this divergence field provides the numerical basis for both the predictive phase-lead analysis in Section~\ref{sec:app_radial_flow} and the entropy analysis presented later in Section~\ref{sec:app_source_coding}. Therefore, we omit these results here for brevity.

Furthermore, the theoretical Centroid Drift $\bm{V}_C(t)$ is shown to track the ensemble average velocity of the particle cloud with zero phase lag, validating the global conservation law in Proposition~\ref{Prop_Centroid_Drift}. This kinematic fidelity forms the theoretical foundation for the global resource guidance strategy demonstrably superior to gradient-based methods, as analyzed in Section~\ref{sec:app_centroid}. Hence, we also omit this result here for brevity.

\begin{figure}[!b]
	\vspace*{-6mm}
	\centering
	\includegraphics[width=0.9\columnwidth]{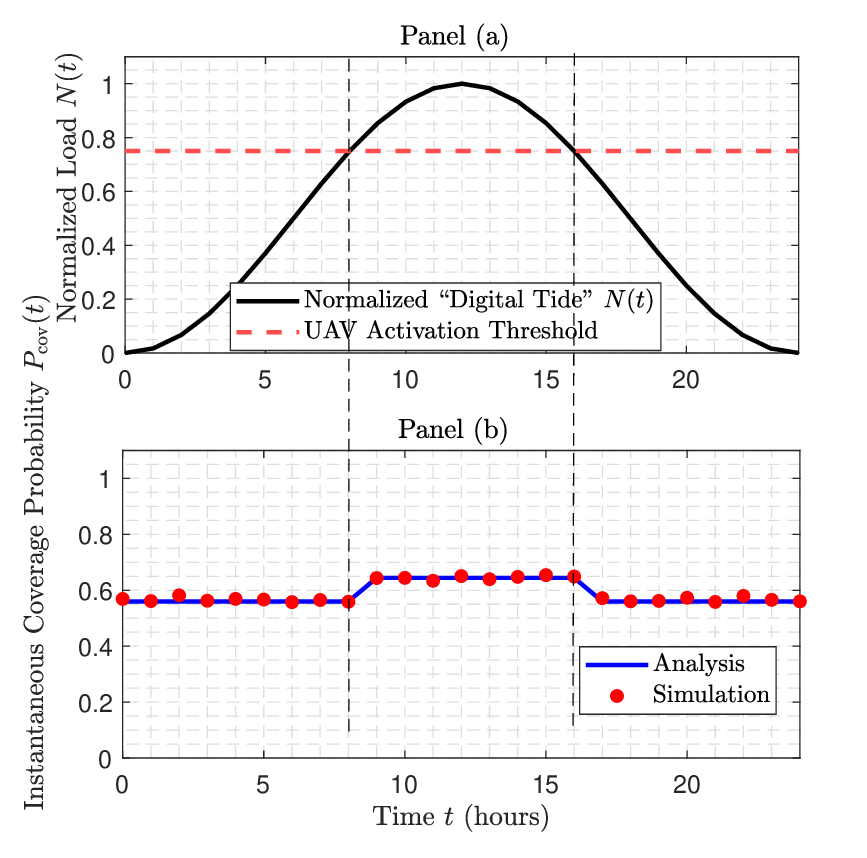}
	\vspace*{-6mm}
	\caption{Validation of dynamic network performance (Propositions \ref{Prop_Assoc_Prob} \& \ref{Prop_Cov_Prob}).}
	\label{fig:pcov_validation} % Fig.6
	\vspace*{-1mm}
\end{figure}

\subsection{Validation of Dynamic Network Performance}\label{sec:sim_performance} % S6.5

Finally, we validate the end-to-end coupling between the continuous intensity field and the discrete infrastructure performance (Propositions~\ref{Prop_Assoc_Prob} \& \ref{Prop_Cov_Prob}). To establish analytically tractable benchmarks, we instantiate the framework using a multi-tier HPPP model, with closed-form expressions for association probability $A_k(t)$ and conditional coverage $P_c^{(k)}(t)$ derived in Appendix~\ref{app:performance_derivation}.

Fig.~\ref{fig:pcov_validation} compares the analytical instantaneous coverage probability $P_{\mathrm{cov}}(t)$ against Monte Carlo snapshots. The results demonstrate that the framework accurately captures the time-varying performance evolution, correctly identifying the transition from the static baseline to the dynamic equilibrium driven by topological reconfiguration. The exact agreement verifies the validity of the adiabatic coupling between the continuous macroscopic intensity and the discrete Poisson infrastructure.

%======================================================================
%=====================Section VII======================================
%======================================================================

\section{Theoretical Insights and Asymptotic Analysis}\label{sec:applications} % S7

\begin{figure}[!b]
	\vspace*{-6mm}
	\centering
	\includegraphics[width=\columnwidth]{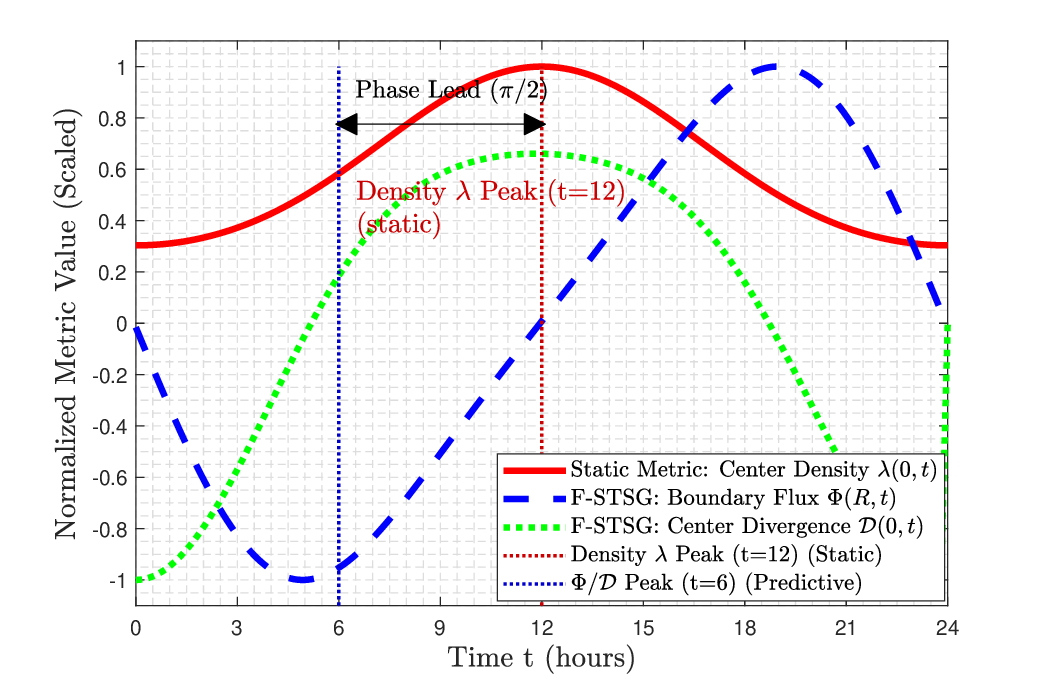}
	\vspace*{-6mm}
	\caption{Temporal causality analysis: the fundamental phase-lead property of the Information Flux ($\bm{J}$) and Congestion Divergence ($\mathcal{D}$) relative to the scalar density ($\lambda$).}
	\label{fig:handover} % Fig.7
	\vspace*{-2mm}
\end{figure}

Having established the numerical robustness and asymptotic consistency of the F-STSG framework, we characterize its utility in deriving fundamental limits and operational insights for non-stationary networks. The macroscopic mobility metrics developed in Section~\ref{sec:fluid_metrics} provide a predictive analytical lens for network orchestration, capturing dynamics that are fundamentally inaccessible via conventional static SG. We demonstrate the theoretical power of this field-theoretic approach through five key insights:
{\color{black}i)~The operational utility of the kinematic phase-lead properties of Information Flux and Congestion Divergence for predictive orchestration;}
ii)~Global spatial resource guidance, demonstrating that tracking the Network Centroid guarantees transport cost minimization compared to reactive gradient-based strategies;
iii)~Information-theoretic source-channel matching to quantify and minimize the entropic cost incurred by control loop latency;
{\color{black}iv)~Asymptotic energy-capacity scaling laws under quadratic-overhead approximations that define the thermodynamic limits of densification independent of physical layer parameters;} and
v)~The fundamental source-channel duality of network mobility, establishing the equivalence between macroscopic flow divergence and topological entropy production.

\subsection{Temporal Causality and Phase-Lead Properties}\label{sec:app_radial_flow} % S7.1

We first analyze the temporal response characteristics of the derived F-STSG operators. Consider the canonical non-separable intensity model, evolving under a periodic diurnal cycle with period $T=24$\,h. While traditional state-based analysis focuses on the scalar intensity $\lambda(r,t)$, our framework introduces vector-based kinematic operators.

Fig.~\ref{fig:handover} illustrates the normalized temporal evolution of the scalar density $\lambda$ versus the derived Information Flux $\bm{J}$ (Proposition~\ref{Prop_Flux_Integral}) and Congestion Divergence $\mathcal{D}$ (Proposition~\ref{Prop_Material_Derivative}).
The analysis reveals a fundamental phase-lead property:
\begin{itemize}
	\item \textit{Zero-Crossing at Peak Load:} At $t=12$\,h, the scalar density reaches its global maximum, where $\partial_t \lambda = 0$. Precisely at this instant, the Information Flux crosses zero. This confirms that $\bm{J}$ acts as the kinematic driver of the state: a vanishing flux is the necessary condition for a stationary point in the load.
	\item \textit{Predictive Causality:} The flux $\bm{J}$ and divergence $\mathcal{D}$ reach their extrema at $t=6$\,h, exhibiting a $T/4$ phase lead relative to the density peak.
	Mathematically, this arises because the flux captures the first-order temporal derivative of the field, satisfying $\nabla \cdot \bm{J} = -\partial_t \lambda$.
\end{itemize}

This phase shift constitutes a deterministic predictive window. Unlike reactive metrics that saturate strictly at the moment of congestion, the F-STSG operators provide sufficient statistics for future topological states, enabling control systems to respond to the cause of congestion—advection—rather than its effect, accumulation.

	\begin{figure}[!b]
	\vspace*{-5mm}
	\centering
	\includegraphics[width=\columnwidth]{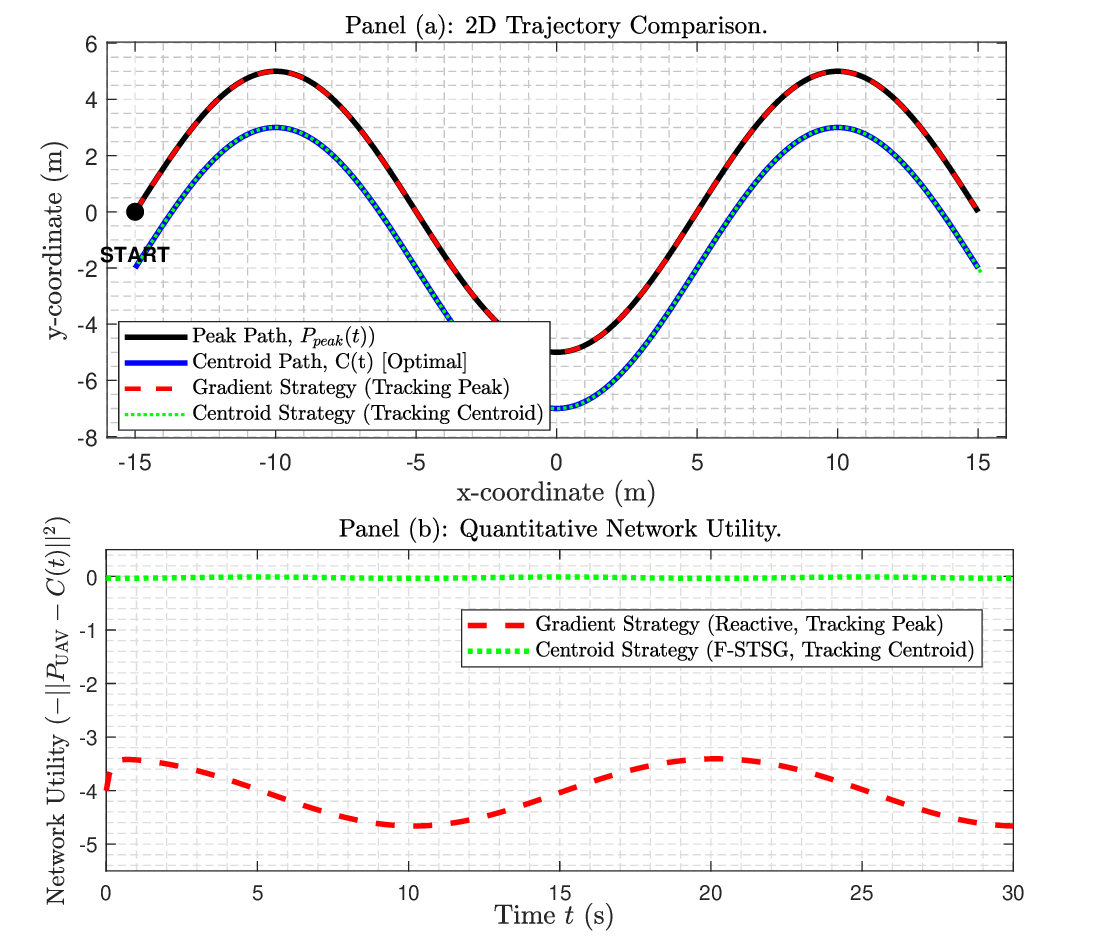}
	\vspace*{-8mm}
	\caption{Spatial optimality analysis: comparative trajectory dynamics in a non-symmetric topology.}
	\label{fig:guidance_application} % Fig.8
	\vspace*{-1mm}
\end{figure}

\subsection{Spatial Optimality and Gradient Dynamics}\label{sec:app_centroid} % S7.2

Beyond temporal prediction, the F-STSG framework offers rigorous bounds on spatial resource orchestration. This perspective is consistent with flux-aware UAV logistics provisioning, where mobile infrastructure is guided by the evolution of the demand field rather than by static hotspot snapshots~\cite{DongDigitalTidesTMC}. We analyze the trajectory optimization problem for a mobile service point, such as a drone, seeking to maximize network utility in a non-symmetric, translating ``comet'' topology, as illustrated in Fig.~\ref{fig:guidance_application}, where the network utility is defined as the negative quadratic transport cost, $U(t) \triangleq -\|\bm{P}_{\text{UAV}}(t) - \bm{C}(t)\|^2$. Maximizing this utility is mathematically equivalent to minimizing the Fr\'{e}chet variance of the system \cite{Villani2003}, thereby ensuring the global minimization of the aggregate path loss.
Two distinct tracking strategies are compared analytically:
\begin{itemize}
	\item \textit{Local Gradient Ascent (Reactive Strategy):} The resource tracks the instantaneous density peak $\bm{P}_{\text{peak}}(t) = \arg \max \lambda(\bm{x},t)$. This represents the standard greedy approach.
	\item \textit{Centroid Dynamics (Proactive):} The resource tracks the Network Centroid $\bm{C}(t)$ given in Definition~\ref{Definition5}, driven by the drift velocity $\bm{V}_C(t)$ derived in Proposition~\ref{Prop_Centroid_Drift}.
\end{itemize}

As shown in Fig.~\ref{fig:guidance_application}(b), the gradient-based strategy suffers from a persistent optimality gap. This is because $\bm{P}_{\text{peak}}$ is a local statistic of the scalar field, which fails to capture the global moment distribution of the probability mass.
In contrast, the F-STSG metric $\bm{C}(t)$ minimizes the Fréchet expectation of the squared distance to the user population.
Consequently, tracking the centroid via $\bm{V}_C(t)$ guarantees the global minimization of the transport cost functional, proving that the macroscopic velocity field $\bm{v}(\bm{x},t)$ provides the optimal control law for global coverage maximization.

\begin{figure}[!b]
	\vspace*{-5mm}
	\centering
	\includegraphics[width=1\columnwidth]{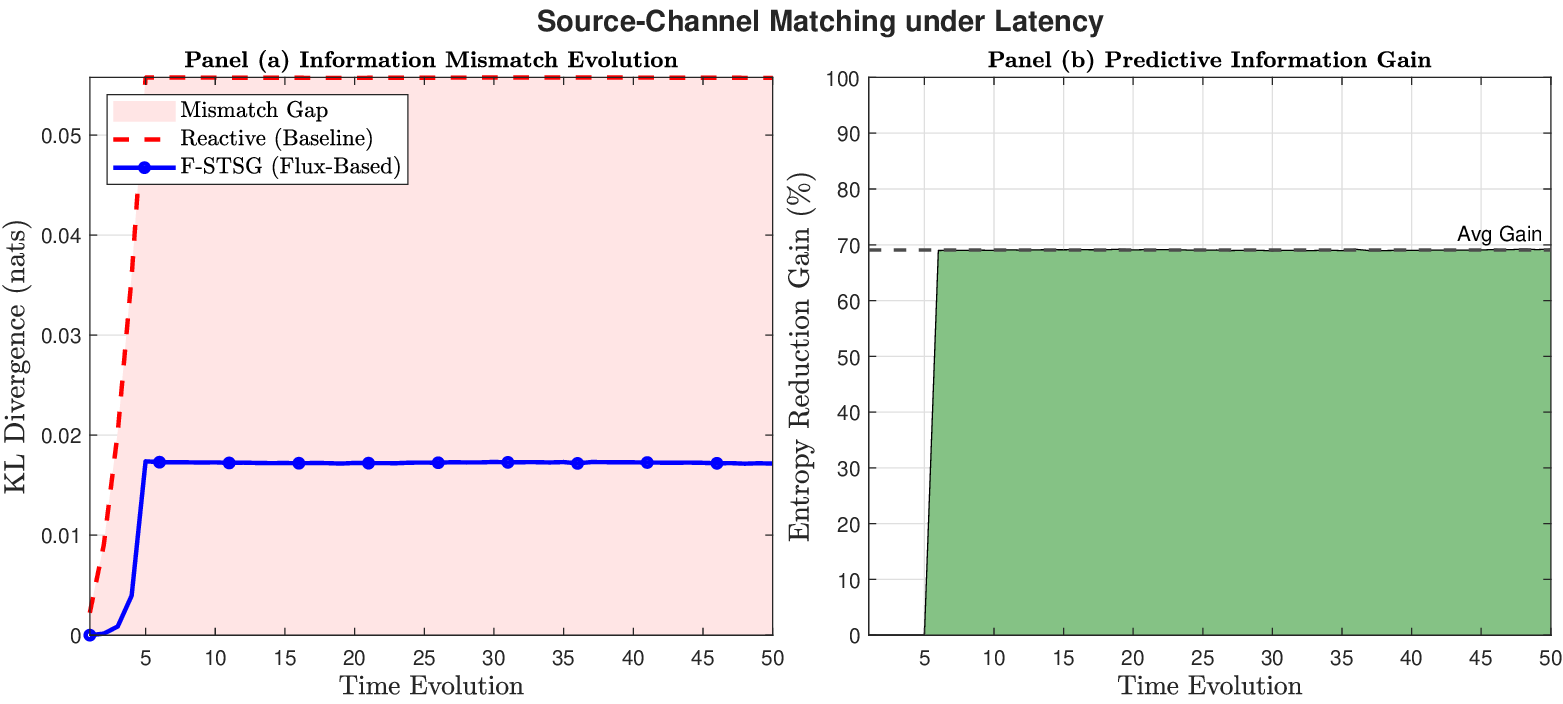} 
	\vspace*{-7mm}
	\caption{Source-channel matching performance under latency ($\tau=5$). (a) Time evolution of KL divergence for the reactive baseline and F-STSG. (b) Percentage of entropy reduction gain achieved by the F-STSG strategy.}
	\label{fig:entropy_validation} % Fig.9
	\vspace*{-1mm}
\end{figure}

\subsection{Source-Channel Matching: Minimizing Entropic Cost of Latency}\label{sec:app_kl_validation} % S7.3

This section evaluates the F-STSG framework from the perspective of source-channel matching, shifting the focus from geometric tracking error to information-theoretic discrepancy. We consider the problem of matching a controlled resource distribution $\mu(\bm{x},t)$ to a non-stationary demand distribution $\lambda(\bm{x},t)$ subject to a control loop latency $\tau$. This formulation provides a numerical verification of the source-channel duality discussed in Section~\ref{sec:variants}.

\subsubsection{Experimental Setup and Metric}
A non-symmetric rotating topology, illustrated in Fig.~\ref{fig:entropy_validation}(a), serves as the time-varying information source in our evaluation. Due to the feedback latency $\tau$, the available CSI is outdated. The reactive strategy (baseline) aligns resources with the lagged observation: $\mu_{\text{react}}(\bm{x},t) = \lambda(\bm{x}, t-\tau)$. In contrast, the F-STSG strategy utilizes the Information Flux $\bm{J}$ to compensate for this advective aging. By invoking the continuity equation, we perform a first-order predictive shift:
\begin{equation} % eq.53
	\mu_{\text{pro}}(\bm{x},t) \approx \lambda(\bm{x}, t-\tau) - \tau \nabla \cdot \bm{J}(\bm{x}, t-\tau).
\end{equation}
The matching efficiency is quantified by the Kullback-Leibler (KL) divergence, $D_{\mathrm{KL}}(\lambda_t || \mu_t)$, which measures the instantaneous information loss (in nats) incurred by the temporal decorrelation of the topology.

\subsubsection{Results and Information Gain}
Fig.~\ref{fig:entropy_validation} illustrates the evolution of the information mismatch. As shown in Panel (a), the reactive strategy yields a persistent entropy penalty of approximately $0.055$\,nats. This divergence quantifies the irreducible uncertainty arising from the spatial lag between stale estimates and the true distribution. By contrast, the F-STSG strategy suppresses the mismatch to a steady-state value of $0.019$\,nats, demonstrating the effectiveness of flux-based prediction in maintaining distribution alignment.

As quantified in Panel (b), the flux-based prediction maintains a stable entropy reduction gain of approximately 65\%. This result confirms that the Information Flux $\bm{J}$ serves as an effective kinematic predictor, translating the outdated measure along the manifold of the Digital Tide. By minimizing this mismatch, the framework reduces the thermodynamic cost of network control and preserves the mutual information between the infrastructure and the mobile demand field.

%======================================================================
%=========================Section VII-C  =============================
%======================================================================

\subsection{Asymptotic Energy-Capacity Scaling Laws}\label{sec:app_optimization} % S7.4

Finally, we utilize the F-STSG framework to derive fundamental scaling limits \cite{Gupta2000} for network densification. We consider the trade-off between ASE and network energy efficiency (EE).

\subsubsection{Problem Formulation}
The activation of the on-demand tier increases the transmitter density $\lambda$.
From an information-theoretic perspective, the ASE $\mathcal{T}(\lambda)$ scales linearly with density in the interference-limited regime, obeying $\mathcal{T} \sim \lambda \bar{R}_{\infty}$.
However, the energetic cost of densification is non-linear.
Let $\Omega(\lambda)$ be a generalized convex cost functional representing the total network power density.
We model this cost based on the rigorous finite-rate feedback theory established by Jindal \cite{Jindal2006}.
In a dense interference network with $K \propto \lambda$ active nodes, achieving the maximal degrees of freedom (DoF) requires suppressing interference to the noise floor.
Crucially, maintaining the full multiplexing gain requires the CSI feedback precision per link to scale linearly with the SIR, yielding $R_{\text{fb}} \propto K$, where $R_{\text{fb}}$ denotes the required rate of CSI feedback (in bits per user) necessary to maintain a bounded rate loss relative to the perfect CSI capacity.
Aggregated over $K$ users, the total network coordination overhead scales as $\mathcal{O}(K^2) \propto \lambda^2$.
Thus, the cost functional $\Omega(\lambda)$ admits a Maclaurin series expansion with a dominant second-order term:
\begin{equation} % eq.54
	\Omega(\lambda) = \underbrace{c_0}_{\text{Static}} + \underbrace{c_1 \lambda}_{\text{Transmission}} + \underbrace{\kappa \lambda^2}_{\text{Coordination}} + \mathcal{O}(\lambda^3).
\end{equation}
The network EE is defined as $\eta_{\mathrm{EE}}(\lambda) = \mathcal{T}(\lambda)/\Omega(\lambda)$.

{\color{black}\subsubsection{Derivation of the Asymptotic Scaling Law}
	Rather than assuming a specific heuristic power model, we analyze the asymptotic behavior of $\Omega(\lambda)$ based on the topology of internode interactions.
	\begin{Proposition}\label{Prop_Scaling} %[Scaling Law] % Pro.7
		Under the quadratic-overhead approximation where coordination cost is dominated by pairwise interactions ($\Omega''(\lambda) \approx 2\kappa > 0$) and higher-order terms are negligible ($\mathcal{O}(\lambda^3) \to 0$), the energy-optimal node density $\lambda^*$ obeys the asymptotic scaling law:
	\begin{equation} % eq.55
		\lambda^* \approx \sqrt{\frac{\mathcal{P}_{\mathrm{static}}}{\kappa}},
		\label{eq:scaling_law_final}
	\end{equation}
	where $\mathcal{P}_{\mathrm{static}} = \Omega(0)$ is the baseline power and $\kappa = \frac{1}{2}\Omega''(0)$ is the effective coordination coefficient.
\end{Proposition}}

\begin{proof}
  See Appendix~\ref{ApendixI}.
\end{proof}

{\color{black}\begin{remark} % Rem.7
	Proposition~\ref{Prop_Scaling} reveals two fundamental design principles:
	\begin{itemize}
		\item \textit{Structural Invariance via Order Cancellation:}
		Crucially, the linear coefficient $c_1$ (dynamic transmission power) is mathematically absent from the optimality condition. This implies that within the quadratic regime, the densification limit is intrinsic to the topology of coordination and decoupled from propagation parameters. We note, however, that this exact cancellation strictly relies on the quadratic truncation; if the optimal density shifts into a regime where higher-order coordination costs (e.g., cubic terms from complex multi-node routing) become non-negligible, the order cancellation breaks down and $c_1$ re-enters the optimality condition.
		\item \textit{The Inverse-Square Information Barrier:}
		The quadratic nature of the feedback cost implies a hard thermodynamic limit.
		The scaling $\lambda^* \propto \kappa^{-1/2}$ dictates that to double the optimal network density, the unit coordination cost must be reduced by a factor of four.
		This proves that merely improving hardware efficiency ($\mathcal{P}_{\mathrm{static}}$) yields diminishing returns, and the fundamental bottleneck is the information entropy of coordination ($\kappa$).
	\end{itemize}
\end{remark}}

\begin{figure}[!b]
\vspace*{-4mm}
	\centering
	\includegraphics[width=0.9\columnwidth]{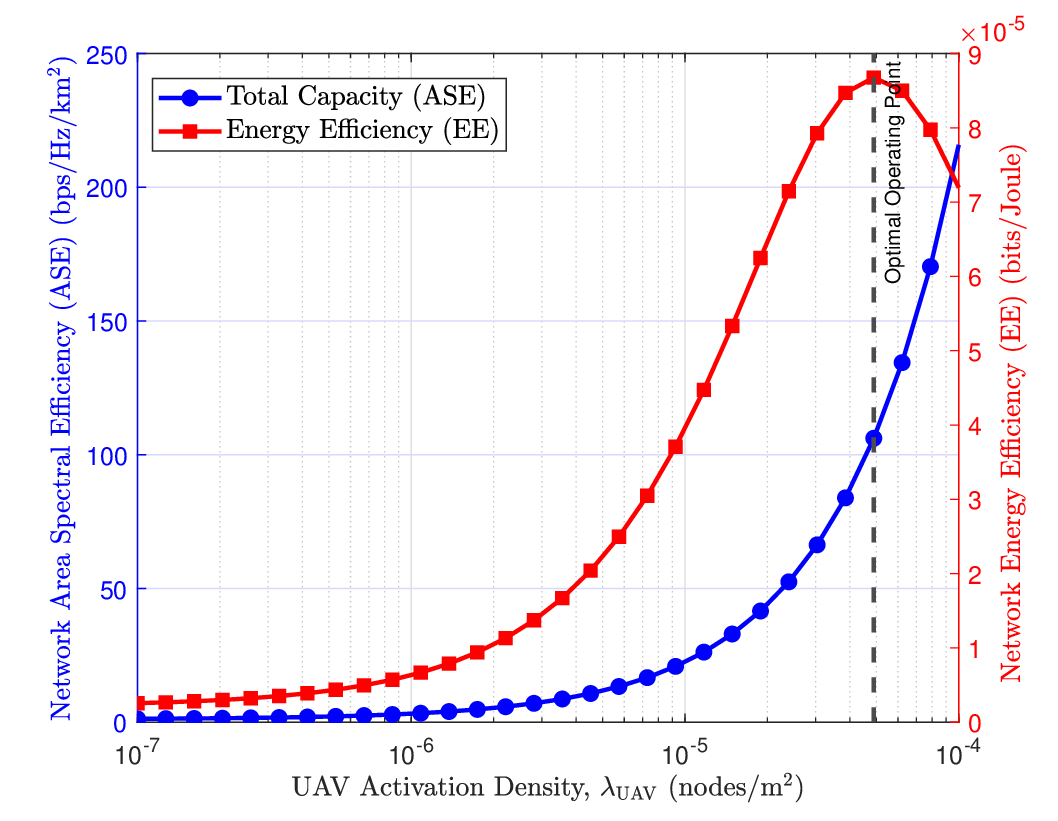}
	\vspace*{-4mm}
	\caption{Energy-capacity scaling analysis: numerical validation of the universal scaling law.}
	\label{fig:optimal_ee_ase} % Fig.10
	\vspace*{-1mm}
\end{figure}

\subsubsection{Thermodynamic Interpretation and Operating Regimes}
Equation \eqref{eq:scaling_law_final} defines the thermodynamic equilibrium point for UDNs.
{\color{black}Numerical validation presented in Fig.~\ref{fig:optimal_ee_ase}, utilizing parameters $\mathcal{P}_{\mathrm{static}}\! \approx\! 100$\,W and $\kappa\! \approx\!
	5 \times 10^7$, confirms this scaling law within the valid bounds of the quadratic approximation and identifies two distinct operating regimes:}
\begin{itemize}
	\item \textit{Noise-Limited Regime ($\lambda \ll \lambda^*$):} The network is sparse. Densification yields positive marginal energy returns as path loss reduction outweighs power costs.
	\item \textit{Coordination-Limited Regime ($\lambda \gg \lambda^*$):} The network is congested by overhead. The quadratic term in $\Omega(\lambda)$ dominates, leading to a rapid decay in energy efficiency.
\end{itemize}

\subsection{The Source-Channel Duality of Mobility}\label{sec:app_source_coding} % S7.5

{\color{black}The F-STSG allows us to elevate the discussion from physical transport to information theory. From the perspective of a network controller, the dynamic topology itself constitutes an evolving information source, while the control signaling infrastructure acts as the communication channel. To maintain optimal resource allocation, the channel capacity must be sufficient to track this non-stationary source. This raises a fundamental conceptual question: \textit{What is the intrinsic information-theoretic footprint of the topological state?} We now prove that the Congestion Divergence $\mathcal{D}$ derived in our framework is rigorously equivalent to the entropy production rate of this topological source, thereby establishing the kinematic foundation for this duality.}

\begin{Proposition}\label{Prop_Entropy} %[Topological Entropy Rate] % Pro.8
	Let $H(t) = - \int_{\mathbb{R}^d} p(\bm{x},t) \ln p(\bm{x},t) \,\mathrm{d}\bm{x}$ be the differential entropy (in nats) of the normalized node distribution.
	In a conservative system governed by the macroscopic velocity field $\bm{v}(\bm{x},t)$, the rate of entropy production is rigorously determined by the spatial expectation of the Congestion Divergence:
	\begin{equation} % eq.56
		\frac{\mathrm{d}H(t)}{\mathrm{d}t} = \mathbb{E}_{\bm{x} \sim p(\cdot,t)} [\mathcal{D}(\bm{x},t)].
		\label{eq:entropy_rate}
	\end{equation}
\end{Proposition}

\begin{proof}
  See Appendix~\ref{ApendixJ}.
\end{proof}

{\color{black}\begin{remark} \label{rem:entropy_interpretation} 
	Proposition~\ref{Prop_Entropy} establishes the continuous foundation for the source-channel duality of network mobility:
	\begin{itemize}
		\item \textit{The Duality Perspective:} The dynamic topology acts as the non-stationary information source, while the control network serves as the channel. Stability dictates that the signaling capacity must accommodate the topological deformation. Here, the Information Flux $\bm{J}$ acts as the kinematic generator of the source's differential entropy.
		\item \textit{Kinematic Information Footprint:} While the differential entropy $H(t)$ can be negative and does not directly equate to a discrete bit rate, its temporal derivative $\frac{\mathrm{d}H}{\mathrm{d}t}$ serves as an intrinsic information footprint. A large absolute magnitude $\left| \frac{\mathrm{d}H}{\mathrm{d}t} \right|$ indicates rapid spatial deformation, signaling severe non-stationarity that stresses the control channel.
		\item \textit{Rate-Distortion Implications:} Establishing the strict operational bit rate $R_{\mathrm{ctrl}}$ requires formalizing a rate-distortion function $R(D)$ with specific spatial quantization bounds. Nevertheless, Proposition~\ref{Prop_Entropy} guarantees that any such minimum coding rate required to bound the tracking error will be fundamentally driven by the statistical moments of the divergence field $\mathcal{D}(\bm{x},t)$.
	\end{itemize}
\end{remark}}

%======================================================================
%===========================Section VIII ==============================
%======================================================================
\section{Theoretical Generalizations and Control Implications}\label{sec:variants} % S8

The F-STSG framework established in Section~\ref{sec:core_theory} was constructed under a set of canonical assumptions—specifically, that the macroscopic transport is conservative ($S=0$) and irrotational ($\nabla \times \bm{v} = \bm{0}$), and that the infrastructure follows a Poisson distribution.
These constraints were imposed to ensure analytical tractability and to isolate the baseline mechanics of the intensity evolution.
However, the core field-theoretic architecture is generic.
In this section, we demonstrate how the framework can be rigorously generalized to model complex physical environments where these constraints are relaxed.

\subsection{Generalized Transport Dynamics}\label{sec:gen_transport} % S8.1

Real-world transport dynamics often exhibit behaviors that violate ideal fluid assumptions. We now formalize the mathematical extensions required to capture non-conservative source dynamics and rotational vorticity.

\subsubsection{Non-Conservative Dynamics (Poisson Formulation)} % S8.1.1)
The baseline derivation assumed that local intensity changes are solely due to advective flux divergence.
However, in the scenarios characterized by bursty demand generation, such as synchronized IoT wake-up cycles, the system is non-conservative.
We generalize the model by introducing a scalar source field $S(\bm{x},t)$.
The governing continuity equation becomes:
\begin{equation} % eq.59
	\frac{\partial \lambda}{\partial t} + \nabla \cdot (\lambda \bm{v}) = S(\bm{x},t).
\end{equation}
Under the minimum-energy potential flow assumption ($\bm{v} = -\nabla \phi$), this transforms the governing equation into a non-homogeneous Poisson equation:
\begin{equation} % eq.60
	\nabla \cdot (\lambda(\bm{x},t) \nabla \phi(\bm{x},t)) = \underbrace{\frac{\partial \lambda}{\partial t} - S(\bm{x},t)}_{\text{Net Source Term } \mathcal{S}_{\mathrm{net}}(\bm{x},t)}.
	\label{eq:poisson_non_conservative}
\end{equation}
This extends the F-STSG to the class of non-homogeneous elliptic PDEs.
The formal solution is expressed via the Green's function $G(\bm{x}, \bm{y})$ of the operator $\nabla \cdot (\lambda \nabla)$:
\begin{equation} % eq.61
	\phi(\bm{x},t) = \int_{\Omega} G(\bm{x}, \bm{y}) \left( \frac{\partial \lambda}{\partial t}(\bm{y},t) - S(\bm{y},t) \right) \, \mathrm{d}\bm{y}.
\end{equation}
This formulation confirms that even with stochastic source terms, the macroscopic velocity field remains mathematically determinate.

\subsubsection{Rotational Dynamics (Vector Potential)} % S8.1.2)
Complex mobility patterns may exhibit vorticity, exemplified by orbital traffic, that cannot be described by a scalar potential alone. To model such rotational flows where $\nabla \times \bm{v} \neq \bm{0}$, we consider the generalized velocity field decomposition:
\begin{equation} % eq.62
	\bm{v}(\bm{x},t) = \bm{v}_{\mathrm{irr}} + \bm{v}_{\mathrm{rot}} = -\nabla \phi(\bm{x},t) + \nabla \times \bm{A}(\bm{x},t),
\end{equation}
where $\bm{A}(\bm{x},t)$ is the vector potential.
In the context of Optimal Transport, the irrotational component $-\nabla \phi$ represents the energy-optimal geodesic path for density reconfiguration. Conversely, the solenoidal component $\bm{v}_{\mathrm{rot}}$ represents non-geodesic circulation. While this component does not contribute to topological compression (divergence is zero) or the mass conservation constraint $\partial_t \lambda$, it strictly increases the kinetic energy of the system. Therefore, $\bm{A}$ captures the ``excess'' kinematic cost of mobility—such as orbital loitering—that exists in reality but is filtered out by the efficiency-driven F-STSG baseline.

\subsection{Microscopic Turbulence and Kinetic Entropy}\label{sec:turbulence} % S8.2

The foundational F-STSG inversion derived in Section~\ref{sec:core_theory} applies the principle of minimum kinetic energy to recover a unique potential field. This solution effectively captures the laminar macroscopic transport, denoted by $\bar{\bm{v}}(\bm{x},t)$. However, in high-mobility regimes, the discrete node distribution exhibits stochastic deviations from this energy-minimizing field. To rigorously characterize this, we invoke the Reynolds decomposition~\cite{Batchelor1967}, splitting the instantaneous velocity field into a coherent mean component and a turbulent fluctuation:
\begin{equation} % eq.63
	\bm{v}(\bm{x},t) = \bar{\bm{v}}(\bm{x},t) + \bm{v}'(\bm{x},t),
\end{equation}
where the fluctuation satisfies $\mathbb{E}[\bm{v}']\! =\! \bm{0}$. Physically, $\bar{\bm{v}}$ represents the coherent macroscopic trend, such as collective handovers, while $\bm{v}'$ captures the local stochastic kinematics. We quantify the intensity of this disorder via the specific turbulent kinetic energy (TKE), defined as $k(\bm{x},t)\! \triangleq\! \frac{1}{2}\mathbb{E}[\|\bm{v}'\|^2]$.

From an information-theoretic perspective, this decomposition reveals the fine-grained structure of the topological entropy rate derived in Proposition~\ref{Prop_Entropy}. While the macroscopic divergence $\nabla \cdot \bar{\bm{v}}$ drives the coherent information generation, the TKE $k(\bm{x},t)$ constitutes an irreducible entropy floor. Consequently, the presence of turbulence degrades the extrapolation fidelity of the reconstructed field. Assuming that the network controller relies solely on the recovered macroscopic potential $\phi$ for state prediction, the tracking error $E(t)$ accumulates according to a turbulent diffusion process. The $L^2$-norm of the density reconstruction error is bounded by the cumulative TKE:
\begin{equation} % eq.64
	E(t) \triangleq \| \lambda - \hat{\lambda} \|_{L^2} \propto \sqrt{\int_0^t \int_{\Omega} k(\bm{x},\tau) \lambda(\bm{x},\tau) \,\mathrm{d}\bm{x} \,\mathrm{d}\tau}.
\end{equation}
This relationship implies that without microscopic correction, the state estimation uncertainty grows with the square root of the integrated turbulent energy.

\begin{remark}\label{rem:topological_temperature} % Rem.10
	In our hydrodynamic framework, the TKE $k(\bm{x},t)$ functions analogously to a thermodynamic temperature. In the low-temperature regime ($k \to 0$), the topology is highly structured and deterministic, allowing for efficient compression of control signaling ($\mathcal{D} \approx 0$). Conversely, in the high-temperature regime ($k \gg 0$), the microscopic vorticity generates significant background entropy. This implies that the universal scaling law in Proposition~\ref{Prop_Scaling} requires modification under turbulence: the effective coordination cost $\kappa$ increases monotonically with $k$, as additional signaling bandwidth is mandated to track the stochastic component $\bm{v}'$ that is orthogonal to the minimum-energy potential flow.
\end{remark}

%=================================================================================

\subsection{Generalized Stochastic Geometry: Structural Correlations}\label{sec:gen_infrastructure} % S8.3

The validation in Section~\ref{sec:simulation} utilized the homogeneous PPP (HPPP) for tractability.
We now extend the performance analysis to general stationary point processes exhibiting structural correlations.

\subsubsection{Pair Correlation Function} % S8.3.1)
For a general infrastructure process $\Phi_{\mathrm{B}}$, the interference statistics are governed by its second-order moment measure. Specifically, we utilize the PCF, denoted as $g(r)$, which represents the normalized probability of finding two nodes separated by distance $r$.
\begin{itemize}
	\item \textit{Cluster Processes:} $g(r) > 1$ for small $r$, capturing the aggregation of small cells in hotspots.
	\item \textit{Repulsive Processes:} $g(r) < 1$ for small $r$, capturing the grid-like regularity of planned macro-cells.
\end{itemize}

\subsubsection{Palm Measure Representation} % S8.3.2)
The interference analysis is rigorously conducted under the reduced Palm distribution.
The generalized Laplace functional of the aggregate interference $I$ is approximated by the PCF-weighted integral (using the Campbell-Mecke theorem approximation):
\begin{align} % eq.65
	\mathcal{L}_I(s) \approx& \exp\Bigg( -2\pi \lambda_{\mathrm{B}} \int_{0}^{\infty} \left(1 - \mathbb{E}_h \left[\frac{1}{1+s P h r^{-\alpha}}\right]\right) \nonumber \\
	 & \times g(r) r \,\mathrm{d}r \Bigg).
\end{align}
By substituting this generalized kernel into the coverage probability derivation (Proposition~\ref{Prop_Cov_Prob}), the F-STSG framework incorporates the impact of infrastructure geometry into the dynamic performance limits.

\subsection{Lagrangian Control Theoretic Extensions}\label{sec:gen_control} % S8.4

Finally, the field-theoretic formulation enables a transition from reactive Eulerian control to predictive Lagrangian control.

\subsubsection{Conservation-Based State Prediction} % S8.4.1)
Traditional control loops suffer from causality lag, reacting only after density accumulates.
Our framework reveals that the future evolution of the intensity field is mathematically encoded in the current divergence of the Information Flux $\bm{J}$.
By coupling the first-order Taylor expansion with the continuity Eeuation, we derive the advection predictor:
\begin{equation} % eq.66
	\hat{\lambda}(\bm{x}, t+\Delta t) \approx \lambda(\bm{x},t) - \Delta t \left( \nabla \cdot \bm{J}(\bm{x},t) - S(\bm{x},t) \right).
\end{equation}
This identity proves that the spatial divergence of the flux is a sufficient statistic for the temporal derivative of the load, enabling zero-latency provisioning.

\subsubsection{Material Derivative Feedback} % S8.4.2)
This leads to a generalized feed-forward control law.
The activation state of on-demand resources, $p_a(\bm{x},t)$, should be defined not as a function of static load $\lambda$, but as a functional of the material derivative:
\begin{equation} % eq.67
	\frac{D\lambda}{Dt} \triangleq \frac{\partial \lambda}{\partial t} + \bm{v} \cdot \nabla \lambda = -\lambda (\nabla \cdot \bm{v}).
\end{equation}
By targeting regions with highly negative material derivative (topological sinks), the network can preemptively allocate resources exactly where the demand field is compressing, stabilizing the quality of service before congestion manifests.

%======================================================================
%==========================Section IX==================================
%======================================================================
\section{Future Research Directions}\label{sec:future_work} % S9

The F-STSG framework establishes a rigorous bridge between continuum mechanics and network information theory, transforming the analysis of mobile networks from static snapshots to dynamic fields.
This foundational work opens three critical theoretical frontiers that warrant further investigation.

\subsection{Analytical Solutions for Generalized Dynamics}\label{S9.1}

While Section~\ref{sec:variants} provided the mathematical formulation for non-conservative and rotational dynamics, deriving closed-form solutions for these generalized equations remains an open challenge.
Future work should focus on solving the non-homogeneous Poisson equation (\ref{eq:poisson_non_conservative}) for specific bursty traffic models, such as synchronized IoT wake-up cycles.
Furthermore, explicitly characterizing the impact of the vector potential $\bm{A}$ (vorticity) on handover signaling overhead—distinct from the scalar potential $\phi$ driving capacity demand—would provide a complete kinematic description of network mobility.
Additionally, integrating the F-STSG driver with advanced point processes, such as determinantal point processes (DPPs) which model repulsion, would refine the interference analysis in ultra-dense, planned networks.

\subsection{Rate-Distortion Theory of Topology}\label{S9.2}
{\color{black}
Building on the source-channel duality perspective established in Section~\ref{sec:app_source_coding}, a fertile avenue lies in defining the fundamental operational limits of topological tracking. Since the network topology acts as a non-stationary information source whose continuous entropy production is driven by the divergence field $\mathbb{E}[\mathcal{D}]$, future research should formally investigate the rate-distortion function $R(D)$ of the mobility field.}
Specifically, given a constraint on the control signaling rate $R < \dot{H}$, what is the optimal quantization scheme that minimizes the Wasserstein-2 distance between the reconstructed and true density fields?
Framing this as a distributed source coding problem with geometric side information will define the ``Information Capacity of Topology,'' providing a rigorous lower bound for the control plane bandwidth in 6G systems.

\subsection{Algorithmic Foundations: Physics-Informed Learning}\label{S9.3}

The third direction focuses on the computational implementation of the framework.
The inverse boundary value problem (Section~\ref{sec:core_theory}) presents a unique opportunity for scientific machine learning.
Rather than purely data-driven prediction, future work could employ physics-informed neural networks (PINNs) that embed the continuity equation directly into the loss function as a regularization term.
This hybrid approach would enable the unsupervised discovery of latent macroscopic velocity fields $\bm{v}(\bm{x},t)$ from sparse, noisy signaling data, such as handover counts, solving the F-STSG inverse problem in complex urban regimes where analytical Green's functions are intractable.

%======================================================================
%=============================Section X ===============================
%======================================================================
\section{Conclusion}\label{sec:conclusion} % S10

This paper has addressed a fundamental theoretical limitation in the modeling of large-scale wireless networks: the inability of the prevailing SG paradigm to analytically capture the macroscopic, collective evolution of network topologies.
To bridge this gap, we have introduced the F-STSG framework.
By modeling the mobile node constellation in the thermodynamic limit, we have moved beyond the static snapshots of traditional analysis to a dynamic continuum description governed by conservation laws.

The core contribution of this work is the resolution of the inverse boundary value problem, establishing the macroscopic velocity field $\bm{v}(\bm{x},t)$ as the unique minimum-energy driver of the intensity evolution.
Leveraging this field-theoretic formulation, we have derived a novel class of kinematic metrics, including the Information Flux $\bm{J}$ and the Congestion Divergence $\mathcal{D}$, which serve as sufficient statistics for predicting topological reconfiguration.
Furthermore, by coupling this continuous transport dynamics with the discrete geometry of the infrastructure, we have characterized the fundamental energy-capacity limits of non-stationary networks.
{\color{black}Specifically, we have derived an asymptotic scaling law $\lambda^*\! \propto\!
	\sqrt{\mathcal{P}_{\mathrm{static}}/\kappa}$ under the quadratic-overhead approximation, revealing an inverse-square information barrier: within this regime, optimal densification is fundamentally limited by the entropy of coordination rather than transmission physics.}
Crucially, we have identified the source-channel duality of mobility, establishing that the entropy production rate of the user field is rigorously determined by the divergence of the Information Flux.
This confirms that topological uncertainty is not merely a kinematic feature but a thermodynamic information source. Consequently, meaningful network orchestration is impossible whenever the control channel capacity fails to match this intrinsic entropy rate.

Ultimately, the F-STSG transforms the network operator's perspective from managing a sequence of static random fields to orchestrating a continuous response to a deterministic hydrodynamic pressure.
It provides the necessary mathematical foundation for the design of next-generation fluid-aware networks.
%======================================================================
%=======================Appendix=======================================
%======================================================================
\appendix

\subsection{Proof of Proposition~\ref{Prop_Optimal_Transport}}\label{app:proof_prop1} % ApA

{\color{black}
\begin{proof}
	We seek to minimize the instantaneous kinetic energy functional subject to the local mass conservation constraint:
	\begin{equation} 
		\min_{\bm{v}} E(t) = \frac{1}{2} \int_{\Omega} \lambda(\bm{x},t) \|\bm{v}(\bm{x},t)\|^2 \,\mathrm{d}\bm{x}, \quad \text{s.t. } \frac{\partial \lambda}{\partial t} + \nabla \cdot (\lambda \bm{v}) = 0.
	\end{equation}
	To solve this constrained variational problem, we introduce a time-dependent scalar Lagrange multiplier $\phi(\bm{x},t)$ to enforce the continuity constraint. The unconstrained Lagrangian functional is given by:
	\begin{equation} 
		\mathcal{L}(\bm{v}, \phi) = \int_{\Omega} \left[ \frac{1}{2} \lambda \|\bm{v}\|^2 - \phi \left( \nabla \cdot (\lambda \bm{v}) + \frac{\partial \lambda}{\partial t} \right) \right] \,\mathrm{d}\bm{x}.
	\end{equation}
	Applying integration by parts to the term $\phi \nabla \cdot (\lambda \bm{v})$ and assuming a no-flux boundary condition ($\lambda \bm{v} \cdot \mathbf{n} = 0$ on $\partial \Omega$), we obtain:
	\begin{equation} 
		\mathcal{L}(\bm{v}, \phi) = \int_{\Omega} \left[ \frac{1}{2} \lambda \|\bm{v}\|^2 + \lambda \bm{v} \cdot \nabla \phi - \phi \frac{\partial \lambda}{\partial t} \right] \,\mathrm{d}\bm{x}.
	\end{equation}
	Taking the variational derivative with respect to the velocity field $\bm{v}$ and setting it to zero yields the first-order optimality condition:
	\begin{equation} 
		\delta_{\bm{v}} \mathcal{L} = \int_{\Omega} \lambda (\bm{v} + \nabla \phi) \cdot \delta \bm{v} \,\mathrm{d}\bm{x} = 0.
	\end{equation}
	Since the intensity field satisfies $\lambda(\bm{x},t) > 0$ strictly inside the domain, this condition must hold for any arbitrary variation $\delta \bm{v}$, which rigorously requires:
	\begin{equation} 
		\bm{v}(\bm{x},t) = -\nabla \phi(\bm{x},t).
	\end{equation}
	Thus, the unique macroscopic velocity field that minimizes the instantaneous kinetic energy while satisfying the continuity equation is strictly an irrotational gradient field.
\end{proof}}

\subsection{Proof of Theorem~\ref{Theorem2}}\label{app:proof_theorem2} % ApB

\begin{proof}
	We substitute the asymptotic expansions for the intensity field $\lambda(\bm{x},t) = \lambda_0 + \epsilon \delta\lambda$ and the potential $\phi(\bm{x},t) = \phi_0 + \epsilon \delta\phi$ into the governing quasi-linear PDE \eqref{eq:conservative_pde}:
	\begin{equation}\label{eqApB1} % eq.69
		\frac{\partial (\lambda_0 + \epsilon \delta\lambda)}{\partial t} = \nabla \cdot \left( (\lambda_0 + \epsilon \delta\lambda) \nabla (\phi_0 + \epsilon \delta\phi) \right).
	\end{equation}
Expanding the terms on the RHS of (\ref{eqApB1}) and grouping them by powers of $\epsilon$:
	\begin{align}\label{eqRHS} % eq.70
		\text{RHS} =& \underbrace{\nabla \cdot (\lambda_0 \nabla \phi_0)}_{\text{Zeroth Order } (\epsilon^0)} 
	  + \epsilon \underbrace{\left( \nabla \cdot (\lambda_0 \nabla \delta\phi) + \nabla \cdot (\delta\lambda \nabla \phi_0) \right)}_{\text{First Order } (\epsilon^1)} \nonumber \\
		& + \epsilon^2 \underbrace{\nabla \cdot (\delta\lambda \nabla \delta\phi)}_{\text{Second Order } (\epsilon^2)}.
	\end{align}
	The LHS of (\ref{eqApB1}) expands linearly as 
	\begin{align}\label{eqLHS} % eq.71
	  \text{LHS} = \frac{\partial \lambda_0}{\partial t} + \epsilon \frac{\partial \delta\lambda}{\partial t}.
	\end{align}
	By definition, the background state satisfies the zeroth-order conservation law identically: 
	\begin{equation} % eq.72
		\frac{\partial \lambda_0}{\partial t} = \nabla \cdot (\lambda_0 \nabla \phi_0).
	\end{equation}
	Subtracting the above zeroth-order term from the both sides, i.e., (\ref{eqLHS}) and  (\ref{eqRHS}), and neglecting the second-order term (linearization assumption $\epsilon \ll 1$) of the RHS (\ref{eqRHS}), we equate the coefficients of the first-order term $\epsilon$:
	\begin{equation} % eq.73
		\frac{\partial (\delta\lambda)}{\partial t} = \nabla \cdot (\lambda_0 \nabla \delta\phi) + \nabla \cdot (\delta\lambda \nabla \phi_0).
	\end{equation}
	Rearranging this equation to isolate the diffusion operator $\nabla \cdot (\lambda_0 \nabla \delta\phi)$ yields the linearized flow potential equation \eqref{eq:linearized_pde}.
	This completes the proof.
\end{proof}

\subsection{Proof of Proposition~\ref{Prop_Flux_Integral}}\label{app:proof_theorem3} % ApC

\begin{proof}
	The net transport rate, $J_{\mathcal{B}}(t)$, crossing the boundary $\mathcal{B}$ is defined as the surface integral of the normal component of the flux density field $\bm{J}(\bm{x},t)$.
	Consider a differential surface element $\mathrm{d}S$ at a point $\bm{x} \in \mathcal{B}$ with outward-pointing unit normal vector $\mathbf{n}$.
	From Definition~\ref{Definition4}, $\bm{J}(\bm{x},t)$ represents the flux density (flow per unit area). The differential flow rate $\mathrm{d}J_{\mathcal{B}}$ through $\mathrm{d}S$ is the projection of the flux vector onto the normal direction:
	\begin{equation} % eq.74
		\mathrm{d}J_{\mathcal{B}}(t) = (\bm{J}(\bm{x},t) \cdot \mathbf{n}) \, \mathrm{d}S.
	\end{equation}
	The total rate $J_{\mathcal{B}}(t)$ is obtained by integrating this differential form over the manifold $\mathcal{B}$:
	\begin{equation} % eq.75
		J_{\mathcal{B}}(t) = \int_{\mathcal{B}} \bm{J}(\bm{x},t) \cdot \mathbf{n} \, \mathrm{d}S.
	\end{equation}
	This completes the proof.
\end{proof}

\subsection{Proof of Proposition~\ref{Prop_Material_Derivative}}\label{app:proof_theorem4} % Ap.D

\begin{proof}
	We start with the conservative continuity equation \eqref{eq:continuity_conservative}:
	\begin{equation} % eq.76
		\frac{\partial \lambda}{\partial t} + \nabla \cdot (\lambda \bm{v}) = 0.
	\end{equation}
	Applying the vector identity for the divergence of a scalar-vector product, $\nabla \cdot (\lambda \bm{v}) = (\nabla \lambda) \cdot \bm{v} + \lambda (\nabla \cdot \bm{v})$, we rewrite the equation as:
	\begin{equation} % eq.77
		\frac{\partial \lambda}{\partial t} + \bm{v} \cdot \nabla \lambda + \lambda (\nabla \cdot \bm{v}) = 0.
	\end{equation}
	Recalling the definition of the material derivative (Lagrangian derivative), $\frac{D\lambda}{Dt} \triangleq \frac{\partial \lambda}{\partial t} + \bm{v} \cdot \nabla \lambda$, the equation simplifies to:
	\begin{equation} % eq.78
		\frac{D\lambda}{Dt} + \lambda (\nabla \cdot \bm{v}) = 0.
	\end{equation}
	Substituting the definition of Congestion Divergence $\mathcal{D} \triangleq \nabla \cdot \bm{v}$, we obtain:
	\begin{equation} % eq.79
		\frac{D\lambda}{Dt} + \lambda \mathcal{D} = 0.
	\end{equation}
	Solving for $\mathcal{D}$ (assuming $\lambda > 0$) yields the result (\ref{eqPro3}).
	This completes the proof.
\end{proof}

\subsection{Proof of Proposition~\ref{Prop_Centroid_Drift}}\label{app:proof_theorem5} % ApE

\begin{proof}
	We differentiate the definition of the Network Centroid $\bm{C}(t) = \frac{1}{N} \int \bm{x} \lambda \,\mathrm{d}\bm{x}$ with respect to time. Since the system is conservative, the total mass $N$ is time-invariant ($\frac{\mathrm{d}N}{\mathrm{d}t} = 0$). Thus:
	\begin{equation} % eq.80
		N \bm{V}_C(t) = N \frac{\mathrm{d}\bm{C}}{\mathrm{d}t} = \int_{\mathbb{R}^d} \bm{x} \frac{\partial \lambda(\bm{x},t)}{\partial t} \,\mathrm{d}\bm{x}.
	\end{equation}
	Substituting the continuity equation $\frac{\partial \lambda}{\partial t} = - \nabla \cdot (\lambda \bm{v})$:
	\begin{equation} % eq.81
		N \bm{V}_C(t) = - \int_{\mathbb{R}^d} \bm{x} \left( \nabla \cdot (\lambda \bm{v}) \right) \,\mathrm{d}\bm{x}.
	\end{equation}
	We apply integration by parts. For the $i$-th component:
	\begin{equation} % eq.82
		\int_{\mathbb{R}^d}\!\! x_i (\nabla \cdot (\lambda \bm{v})) \,\mathrm{d}\bm{x} = \underbrace{\left( x_i (\lambda \bm{v}) \cdot \mathbf{n} \right)_{\partial \mathbb{R}^d}}_{0} -\! \int_{\mathbb{R}^d}\! (\nabla x_i) \cdot (\lambda \bm{v}) \,\mathrm{d}\bm{x},
	\end{equation}
	where the boundary term vanishes assuming the flux decays to zero at infinity. Since $\nabla x_i = \mathbf{e}_i$ (the basis vector), the integral becomes $-\int \lambda v_i \,\mathrm{d}\bm{x}$.
	Substituting this back into the vector equation:
	\begin{align} % eq.83
		N \bm{V}_C(t) =& - \left( - \int_{\mathbb{R}^d} \lambda(\bm{x},t) \bm{v}(\bm{x},t) \,\mathrm{d}\bm{x} \right) \nonumber \\ =& \int_{\mathbb{R}^d} \lambda(\bm{x},t) \bm{v}(\bm{x},t) \,\mathrm{d}\bm{x}.
	\end{align}
	Dividing by $N$ yields the intensity-weighted average velocity.
	This completes the proof.
\end{proof}

\subsection{Proof of Proposition~\ref{Prop_Assoc_Prob}}\label{app:proof_theorem6} % ApF

\begin{proof}
	Let $E_k$ denote the event that a typical mobile node associates with tier $k$. The global association probability is $A_k(t) = \mathbb{P}(E_k)$.
	Let $\bm{X}$ be the random location of the mobile node, distributed according to the instantaneous spatial PDF $p(\bm{x},t) = \lambda(\bm{x},t) / N(t)$.
	By the law of total probability:
	\begin{equation} % eq.84
		A_k(t) = \mathbb{E}_{\bm{X}} \left[ \mathbb{P}(E_k \mid \bm{X}) \right] = \int_{\mathbb{R}^d} \mathcal{A}_k(\bm{x},t) p(\bm{x},t) \,\mathrm{d}\bm{x},
	\end{equation}
	where $\mathcal{A}_k(\bm{x},t) \triangleq \mathbb{P}(E_k \mid \bm{X} = \bm{x})$ is the location-specific association probability.
	Given the association rule in \eqref{eq:assoc_rule}, the event $E_k \mid \bm{x}$ corresponds to the condition that the serving node $\bm{X}^*$ belongs to the active tier-$k$ set $\Psi_k(t)$:
	\begin{equation} % eq.85
		\mathcal{A}_k(\bm{x},t) = \mathbb{P}\left( \mathcal{K}(\bm{X}^*) = k \mid \bm{x} \right).
	\end{equation}
	This probability is uniquely determined by the SG of the infrastructure processes.
\end{proof}

\subsection{Proof of Proposition~\ref{Prop_Cov_Prob}}\label{app:proof_theorem7} % ApG

\begin{proof}
	Let $E_{\mathrm{cov}}$ be the event that a typical mobile node achieves the target SINR. This event is the disjoint union of coverage events across all tiers.
	Let $E_k$ be the association event as defined in Appendix~\ref{app:proof_theorem6}, and $C_k$ be the conditional coverage event given $E_k$ (i.e., $\text{SINR}_k > \gamma_k$).
	We have:
	\begin{equation} % eq.86
		P_{\mathrm{cov}}(t)\! = \!\mathbb{P}(E_{\mathrm{cov}}) \!= \!\! \sum_{k=1}^K \mathbb{P}(C_k \cap E_k) \!=\!\! \sum_{k=1}^K \mathbb{P}(C_k \mid E_k) \mathbb{P}(E_k).
	\end{equation}
	Recognizing $\mathbb{P}(E_k) = A_k(t)$ and defining the conditional coverage probability as $P_{\textrm{c}}^{(k)}(t) \triangleq \mathbb{P}(\text{SINR} > \gamma_k \mid E_k)$ yields \eqref{eq:pcov_instantaneous}.
\end{proof}

\subsection{Analytical Derivation for Canonical Poisson Instantiation}\label{app:performance_derivation} % ApH

This appendix details the derivation of the closed-form expressions employed in the simulation validation (Section~\ref{sec:sim_performance}). Invoking the adiabatic assumption (Assumption~\ref{Assump:Adiabatic}), we treat the instantaneous node distribution as a stationary spatial process. While the general F-STSG framework supports arbitrary intensity fields $\lambda(\bm{x},t)$, analytical tractability requires specifying the underlying point processes. Consistent with the simulation setup, we adopt the following canonical assumptions:
\begin{itemize}
	\item The baseline infrastructure $\Phi_k$ is a HPPP with density $\lambda_k$.
	\item The dynamic activation is an independent thinning process driven by the macroscopic intensity, resulting in an active node process $\Psi_k(t)$ with time-varying density $\lambda_{a,k}(t) = p_{a,k}(t)\lambda_k$.
	\item The channel follows Rayleigh fading ($h \sim \exp(1)$) with path loss exponent $\alpha > 2$.
	\item Mobile nodes associate based on the maximum biased mean received power (MBMRP) rule.
\end{itemize}

\subsubsection{Instantaneous Association Probability (Proposition~\ref{Prop_Assoc_Prob})}\label{app:assoc_prob}
The instantaneous association probability $A_k(t)$ is the spatial average of the local association preference. Due to the stationarity of the HPPP infrastructure, the local preference is translation invariant. Thus, $A_k(t)$ is simply the probability that a tier-$k$ node provides the maximum biased received power:
\begin{equation} % eq.87
	A_k(t)\! =\! \mathbb{P}\! \left( \max_{\bm{X}_i \in \Psi_k(t)} \frac{P_k \mathcal{W}_k}{\|\bm{X}_i\|^{\alpha}} > \max_{j \neq k} \max_{\bm{X}_j \in \Psi_j(t)} \frac{P_j \mathcal{W}_j}{\|\bm{X}_j\|^{\alpha}}\! \right)\! .\!
\end{equation}
Let $R_k\! =\! \min_{\bm{X}_i \in \Psi_k(t)}\! \|\bm{X}_i\|$ be the distance to the nearest tier-$k$ node. For an HPPP, the PDF of $R_k$ is $f_{R_k}(r)\! =\! 2\pi \lambda_{a,k}(t) r \exp\big(-\pi \lambda_{a,k}(t) r^2\big)$. The association event is equivalent to $P_k \mathcal{W}_k R_k^{-\alpha} > P_j \mathcal{W}_j R_j^{-\alpha}$ for all $j$.
Solving this standard SG problem yields the closed-form expression used in our simulation:
\begin{equation} % eq.88
	A_k(t) = \frac{\lambda_{a,k}(t) (P_k \mathcal{W}_k)^{2/\alpha}}{\sum_{j=1}^K \lambda_{a,j}(t) (P_j \mathcal{W}_j)^{2/\alpha}}.
	\label{eq:app_assoc_final}
\end{equation}

\subsubsection{Instantaneous Conditional Coverage Probability (Proposition~\ref{Prop_Cov_Prob})}\label{app:cov_prob}
We derive the conditional probability $P_{\textrm{c}}^{(k)}(t) \triangleq \mathbb{P}(\text{SINR}_k > \gamma_k \mid E_k)$.
The SINR is given by $\text{SINR}_k = \frac{P_k h r^{-\alpha}}{I_{\text{total}} + \sigma^2}$. In the interference-limited regime ($\sigma^2 \to 0$) with Rayleigh fading ($h \sim \exp(1)$), the coverage probability is:
\begin{align} % eq.89
	P_{\textrm{c}}^{(k)}(t) &= \int_0^\infty \mathbb{P}\left( h > \frac{\gamma_k r^{\alpha}}{P_k} I_{\text{total}} \right) f_{R_0 \mid E_k}(r) \,\mathrm{d}r \nonumber \\
	&= \int_0^\infty \mathbb{E}_{I} \left[ \exp\left( - \frac{\gamma_k r^{\alpha}}{P_k} I_{\text{total}} \right) \right] f_{R_0 \mid E_k}(r) \,\mathrm{d}r \nonumber \\
	&= \int_0^\infty \left( \prod_{j=1}^K \mathcal{L}_{I_j}\left(s = \frac{\gamma_k r^{\alpha}}{P_k}\right) \right) f_{R_0 \mid E_k}(r) \,\mathrm{d}r.
	\label{eq:cov_integral_form}
\end{align}
We now explicitly derive the two components: $f_{R_0 \mid E_k}(r)$ and $\mathcal{L}_{I_j}(s)$.

\paragraph{Serving Distance Distribution}
Given the association condition $E_k$, the distance $R_0$ to the serving node is distributed according to:
\begin{equation} % eq.90
	f_{R_0 \mid E_k}(r) = \frac{2\pi \lambda_{a,k}(t) r \exp\! \left(\! -\pi r^2 \sum\limits_{j=1}^K \lambda_{a,j}(t) \big(\frac{P_j \mathcal{W}_j}{P_k \mathcal{W}_k}\big)^{2/\alpha}\! \right)}{A_k(t)}.
\end{equation}
Substituting $A_k(t)$ from \eqref{eq:app_assoc_final} simplifies this to:
\begin{equation} % eq.91
	f_{R_0 \mid E_k}(r) = 2\pi \Lambda_{k,\text{assoc}}(t) r \exp\left( -\pi r^2 \Lambda_{k,\text{assoc}}(t) \right),
	\label{eq:serving_dist_pdf}
\end{equation}
where $\Lambda_{k,\text{assoc}}(t) \triangleq \sum_{j=1}^K \lambda_{a,j}(t) \left( \frac{P_j \mathcal{W}_j}{P_k \mathcal{W}_k} \right)^{2/\alpha}$.

\paragraph{Interference Laplace Transform}
The aggregate interference $I_j$ from tier $j$ originates from active nodes $\Psi_j(t)$ distributed in $\mathbb{R}^d$ excluding the exclusion ball $b(0, r C_{j,k})$, where $C_{j,k} = (P_j \mathcal{W}_j / P_k \mathcal{W}_k)^{1/\alpha}$ ensures the association condition is met. The Laplace transform is derived as:
\begin{align} % eq.92
	& \mathcal{L}_{I_j}(s) = \mathbb{E}_{\Psi_j} \left[ \prod_{i: X_i \in \Psi_j} \frac{1}{1 + s P_j \|\bm{X}_i\|^{-\alpha}} \right] \nonumber \\
	& ~ = \exp\! \left(\! -2\pi \lambda_{a,j}(t)\! \int_{r C_{j,k}}^\infty \!\! \left(\! 1 - \frac{1}{1\! +\! s P_j v^{-\alpha}} \right) v \,\mathrm{d}v\! \right)\!\!.
\end{align}
Substituting $s = \frac{\gamma_k r^{\alpha}}{P_k}$ and utilizing the substitution $u = (v/r C_{j,k})^2$, the integral simplifies to:
\begin{equation} % eq.93
	\mathcal{L}_{I_j}(s) = \exp\left( - \pi r^2 \lambda_{a,j}(t) \rho_j(\gamma_k, \alpha) \right),
	\label{eq:laplace_final}
\end{equation}
where $\rho_j(\gamma_k, \alpha) \triangleq (P_j \mathcal{W}_j / P_k \mathcal{W}_k)^{2/\alpha} \int_{1}^{\infty} \frac{\gamma_k}{u^{\alpha/2} + \gamma_k} \mathrm{d}u$.

\paragraph{Final Integration}
Substituting \eqref{eq:serving_dist_pdf} and \eqref{eq:laplace_final} back into \eqref{eq:cov_integral_form}, we obtain:
\begin{align} % eq.94
	& P_{\textrm{c}}^{(k)}(t) = 2\pi \Lambda_{k,\text{assoc}}(t) \int_0^\infty r \nonumber\\
	& ~ \times \exp\! \left(\! -\pi r^2\! \left(\! \Lambda_{k,\text{assoc}}(t) + \sum_{j=1}^K \lambda_{a,j}(t) \rho_j \right)\! \right) \mathrm{d}r.
\end{align}
Noting the integral $\int_0^\infty 2 C_1 r \exp(- (C_1 + C_2) r^2) \mathrm{d}r\! =\! \frac{C_1}{C_1 + C_2}$, the closed-form expression is derived as:
\begin{equation} % eq.95
	P_{\textrm{c}}^{(k)}(t) = \frac{\Lambda_{k,\text{assoc}}(t)}{\Lambda_{k,\text{assoc}}(t) + \sum_{j=1}^K \lambda_{a,j}(t) \rho_j(\gamma_k, \alpha)}.
	\label{eq:app_cov_final}
\end{equation}

\subsubsection{Overall Instantaneous Coverage Probability}
Combining the results, the total instantaneous coverage probability is:
\begin{equation} % eq.96
	P_{\mathrm{cov}}(t) = \sum_{k=1}^K A_k(t) \cdot P_{\textrm{c}}^{(k)}(t).
	\label{eq:app_pcov_total_final}
\end{equation}
This derivation explicitly links the time-varying activation densities $\lambda_{a,k}(t)$ to the precise microscopic performance metrics.

\subsection{Proof of Proposition~\ref{Prop_Scaling}}\label{ApendixI}

\begin{proof}
	We analyze the asymptotic behavior of the objective function $\eta_{\mathrm{EE}}(\lambda) \approx \frac{\lambda \bar{R}_{\infty}}{\Omega(\lambda)}$.
	Differentiating $\eta_{\mathrm{EE}}(\lambda)$ with respect to $\lambda$ using the quotient rule:
	\begin{equation} % eq.97
		\frac{\mathrm{d} \eta_{\mathrm{EE}}}{\mathrm{d} \lambda} = \frac{\bar{R}_{\infty}(\Omega(\lambda) - \lambda \Omega'(\lambda))}{(\Omega(\lambda))^2}.
	\end{equation}
	The optimal density is determined by the vanishing of the numerator.
	Truncating the expansion to the second order ($\Omega(\lambda) \approx c_0 + c_1 \lambda + \kappa \lambda^2$), the optimality condition becomes:
	\begin{equation} % eq.98
		(c_0 + c_1 \lambda + \kappa \lambda^2) - \lambda(c_1 + 2\kappa \lambda) \approx 0.
	\end{equation}
	Simplifying this condition reveals a fundamental structural cancellation:
	\begin{equation} % eq.99
		c_0 - \kappa \lambda^2 \approx 0 \implies \lambda^* \approx \sqrt{\frac{c_0}{\kappa}}.
	\end{equation}
	The second derivative at this critical point is $\eta_{\rm EE}''(\lambda^*) < 0$, confirming a global maximum. 
\end{proof}

\subsection{Proof of Proposition~\ref{Prop_Entropy}}\label{ApendixJ}

\begin{proof}
	By definition, the differential entropy is $H(t) = - \int_{\mathbb{R}^d} p \ln p \,\mathrm{d}\bm{x}$. Differentiating it yields:
	\begin{equation} % eq.100
		\frac{\mathrm{d}H}{\mathrm{d}t} = - \int_{\mathbb{R}^d} \frac{\partial p}{\partial t} (1 + \ln p) \,\mathrm{d}\bm{x}.
	\end{equation}
	Substituting the continuity equation $\partial_t p = - \nabla \cdot (p \bm{v})$:
	\begin{equation} % eq.101
		\frac{\mathrm{d}H}{\mathrm{d}t} = \int_{\mathbb{R}^d} \nabla \cdot (p \bm{v}) (1 + \ln p) \,\mathrm{d}\bm{x}.
	\end{equation}
	Applying the Divergence Theorem (assuming the flux $p \bm{v}$ vanishes at infinity) and integration by parts:
	\begin{align} % eq.102
		\frac{\mathrm{d}H}{\mathrm{d}t} &= - \int_{\mathbb{R}^d} (p \bm{v}) \cdot \nabla (1 + \ln p) \,\mathrm{d}\bm{x} \nonumber \\
		&= - \int_{\mathbb{R}^d} p \bm{v} \cdot \frac{\nabla p}{p} \,\mathrm{d}\bm{x} = - \int_{\mathbb{R}^d} \bm{v} \cdot \nabla p \,\mathrm{d}\bm{x}.
	\end{align}
	Using the vector identity $\nabla \cdot (p \bm{v}) = p (\nabla \cdot \bm{v}) + \bm{v} \cdot \nabla p$, we substitute $-\bm{v} \cdot \nabla p = p (\nabla \cdot \bm{v}) - \nabla \cdot (p \bm{v})$. Integrating over the domain, the term $\int \nabla \cdot (p \bm{v}) \mathrm{d}\bm{x}$ vanishes due to mass conservation. Thus:
	\begin{equation} % eq.103
		\frac{\mathrm{d}H}{\mathrm{d}t} = \int_{\mathbb{R}^d} p(\bm{x},t) (\nabla \cdot \bm{v}(\bm{x},t)) \,\mathrm{d}\bm{x}.
	\end{equation}
	Recalling the definitions of $\mathcal{D}(\bm{x},t) \triangleq \nabla \cdot \bm{v}(\bm{x},t)$ and $\mathbb{E}_{\bm{x} \sim p}[\cdot]$, we finally obtain:
	\begin{equation} % eq.104
		\frac{\mathrm{d}H(t)}{\mathrm{d}t} = \int_{\mathbb{R}^d} \mathcal{D}(\bm{x},t) p(\bm{x},t) \,\mathrm{d}\bm{x} \equiv \mathbb{E}_{\bm{x} \sim p(\cdot,t)} [\mathcal{D}(\bm{x},t)].
	\end{equation}
	This completes the proof.
\end{proof}

%\bibliography{IEEEtran}
%\bibliography{reference}	
\small
%\bibliography{IEEEabrv, reference}
%\bibliography{reference}		

\end{document}